\documentclass[conference]{IEEEtran}

\usepackage{subcaption}
\usepackage{algorithm}
\usepackage[noend]{algpseudocode}
\usepackage[noadjust]{cite}
\usepackage{csquotes}
\usepackage{MnSymbol}
\usepackage{multicol}
\usepackage[keys, sets, operators, asymptotics, primitives, complexity, ff, landau, logic]{cryptocode}
\usepackage{mdframed}
\usepackage{todonotes}
\usepackage{pgf}
\usepackage{tikz}
\usetikzlibrary{arrows,automata}
\usepackage{caption}

\algrenewcommand\algorithmicindent{0.5em}%

\newtheorem{theorem}{Theorem}[section]

\newtheorem{definition}[theorem]{Definition}

\newtheorem{lemma}[theorem]{Lemma}

\newcommand\heartdisease{HDI}
\newcommand\housing{HOU}

\newcommand\spambase{SPA}
\newcommand\artificial{ART}

\newcommand\ptxt{\mathsf{m}} 
\newcommand\ctxt{\mathsf{c}} 
\newcommand{\ctxtrep}[1]{\lsem #1 \rsem}
\newcommand{\func}[2]{\mathcal{#1}_{\mathsf{#2}}} 
\newcommand\ek{\mathsf{ek}}
\newcommand\secprm{\lambda}
\newcommand\inputlen{\mu}
\newcommand\tfun{\mathsf{thr}} 
\newcommand\afun{\mathsf{att}} 
\newcommand\cfun{\mathsf{lab}} 

\newcommand\path{\mathsf{path}} %
\newcommand\result{\mathsf{result}} %
\newcommand\from{\mathsf{from}} %
\newcommand\tosf{\mathsf{to}} %
\newcommand\leftsf{\mathsf{left}} %
\newcommand\rightsf{\mathsf{right}} %
\newcommand\midsf{\mathsf{mid}} %
\newcommand\up{\mathsf{up}} %
\newcommand\low{\mathsf{low}} %
\newcommand\level{\mathsf{level}} %
\newcommand\currLvl{\mathsf{currLvl}} %
\newcommand\destLvl{\mathsf{destLvl}} %
\newcommand\cost{\mathsf{cost}} %

\newcommand\helibsmall{\text{HElib}_\text{small}} 
\newcommand\helibmed{\text{HElib}_\text{med}} 
\newcommand\helibbig{\text{HElib}_\text{big}} 
\newcommand\helibint{\text{HElib}_\text{int}} 
\newcommand\tfheIZB{\text{TFHE}_\text{128}}

\newcommand{\prot}[1]{\Pi_{\mathsf{#1}}} 
\newcommand{\simul}[2]{\simulator_{\mathsf{#1}}^{\mathsf{\MakeLowercase{#2}}}} 
\newcommand{\view}[2]{\mathsf{View}_{\mathsf{#1}}^{#2}} 

\newcommand{\bitrep}[1]{#1^{\mathsf{b}}}
\newcommand{\bitrepx}[2]{#1^{(#2)^{\mathsf{b}}}}

\newcommand\sheadd{\textsc{SheAdd}} 
\newcommand\shemul{\textsc{SheMult}} 
\newcommand\shecmp{\textsc{SheCmp}} 
\newcommand\shefadder{\textsc{SheFadder}} 
\newcommand\sheequ{\textsc{SheEqual}} 
\newcommand\slotnb{s}
\newcommand\slotct{l}
\newcommand\nblevel{L}
\newcommand\poldegree{N}

\newcommand\node{\mathsf{Node}}
\newcommand\rootnode{\mathsf{root}}
\newcommand\thr{\mathsf{threshold}}
\newcommand\aindex{\mathsf{aIndex}}
\newcommand\pnode{\mathsf{parent}}
\newcommand\lnode{\mathsf{left}}
\newcommand\rnode{\mathsf{right}}
\newcommand\cmp{\mathsf{cmp}}
\newcommand\mdag{\mathsf{dag}}
\newcommand\clabel{\mathsf{cLabel}}
\newcommand\nodeset{\mathcal{D}}
\newcommand\leafset{\mathcal{L}}
\newcommand\classfun{\mathsf{f}_c}
\newcommand\travfun{\mathsf{tr}}
\newcommand\evalnode{\textsc{EvalDnode}}
\newcommand\evalpath{\textsc{EvalPaths}}
\newcommand\evalpathe{\textsc{EvalPathsE}}
\newcommand\evalpathp{\textsc{EvalPathsP}}
\newcommand\evalmul{\textsc{EvalMul}}
\newcommand\addedge{\textsc{AddEdge}}
\newcommand\computedag{\textsc{ComputeDag}}
\newcommand\evalleaves{\textsc{Finalize}}

\newcommand{\enQ}{\mathsf{enqueue}}
\newcommand{\deQ}{\mathsf{dequeue}}

\newcommand{\push}{\mathsf{push}}
\newcommand{\pop}{\mathsf{pop}}
\newcommand{\emptyQ}{\mathsf{empty}}

\algnewcommand\algorithmicforeach{\textbf{for each}}
\algdef{S}[FOR]{ForEach}[1]{\algorithmicforeach\ #1\ \algorithmicdo}
\newcommand\dagcfun{h}
\newcommand\dagfun{g}

\newcommand\pdtebin{\textsc{Pdt-Bin}}
\newcommand\pdteint{\textsc{Pdt-Int}}

\begin{document}

\title{Non-Interactive Private Decision Tree Evaluation}

\author{\IEEEauthorblockN{Anselme Tueno}
\IEEEauthorblockA{SAP SE\\
anselme.kemgne.tueno@sap.com}
\and
\IEEEauthorblockN{Yordan Boev}
\IEEEauthorblockA{SAP SE\\
iordan.boev@gmail.com}
\and
\IEEEauthorblockN{Florian Kerschbaum}
\IEEEauthorblockA{University of Waterloo\\
florian.kerschbaum@uwaterloo.ca}}

\maketitle

\begin{abstract}
Decision trees are a powerful prediction model with many applications in statistics, data mining, and machine learning. In some settings, the model and the data to be classified may contain sensitive information belonging to different parties. In this paper, we, therefore, address the problem of privately evaluating a decision tree on private data. This scenario consists of a server holding a private decision tree model and a client interested in classifying its private attribute vector using the server’s private model. The goal of the computation is to obtain the classification while preserving the privacy of both – the decision tree and the client input. After the computation, the classification result is revealed only to the client, and nothing else is revealed neither to the client nor to the server. Existing privacy-preserving protocols that address this problem use or combine different generic secure multiparty computation approaches resulting in several interactions between the client and the server.
Our goal is to design and implement a novel client-server protocol that delegates the complete tree evaluation to the server while preserving privacy and reducing the overhead. The idea is to use fully (somewhat) homomorphic encryption and evaluate the tree on ciphertexts encrypted under the client's public key. However, since current somewhat homomorphic encryption schemes have high overhead, we combine efficient data representations with different algorithmic optimizations to keep the computational overhead and the communication cost low. As a result, we are able to provide the first non-interactive protocol, that allows the client to delegate the evaluation to the server by sending an encrypted input and receiving only the encryption of the result. Our scheme has only one round and can evaluate a complete tree of depth 10 within seconds.
\end{abstract}


\maketitle

\section{Introduction}

A machine learning process consists of two phases. In the first phase or \emph{learning phase}, a model or classifier is built on a possibly large set of training data. The model can then be used to classify new data.

\paragraph{Setting} Machine learning (ML) classifiers are valuable tools in many areas such as healthcare, finance, spam filtering, intrusion detection, remote diagnosis, etc \cite{WittenFH.2011}. To perform their task, these classifiers often require access to personal sensitive data such as medical or financial records. Therefore, it is crucial to investigate technologies that preserve the privacy of the data, while benefiting from the advantages of ML. On the one hand, the ML model itself may contain sensitive data. For example, a bank that uses a decision tree for credit assessment of its customers may not want to reveal any information about the model. On the other hand, the model may have been built on sensitive data. It is known that white-box and sometimes even black-box access to a ML model allows so-called \emph{model inversion attacks} \cite{WuFJN.2016, FredriksonJR.2015, TramerZJRR.2016}, which can compromise the privacy of the training data. As a result, making the ML model public could violate the privacy of the training data.

\paragraph{Scenario}  In this paper, we therefore address the problem of private decision tree evaluation (PDTE) on private data. This scenario consists of a server holding a private decision tree model and a client wanting to classify its private attribute vector using the server’s private model. The goal of the computation is to obtain the classification while preserving the privacy of both – the decision tree and the client input. After the computation, the classification result is revealed only to the client, and beyond that, nothing further is revealed to neither party. The problem can be solved using any generic secure multiparty computation. There exist specialized solutions that combine different techniques and use the domain knowledge to develop efficient protocols. 

\paragraph{Generic Secure Two-party Computation Solution} Generic secure two-party computation \cite{Yao.1982, Goldreich.2004, CramerDN2001}, such as garbled circuit and secret sharing, can implement PDTE. The idea is to transform the decision tree program into a secure representation that can be evaluated without revealing private data. There exist frameworks such as ObliVM \cite{LiuWNHS.2015} or CBMC-GC \cite{FranzHKSV.2014} that are able to automate the transformation of the plaintext programs, written in a high level programming language, into oblivious programs suitable for secure computation. Their straightforward application to decision tree programs does certainly improve performance over a hand-crafted construction. However, the size of the resulting oblivious program is still proportional to the size of the tree. 
%
%
As a result generic solution are in general inefficient, in particular when the size of the tree is large.

\paragraph{Specialized Solutions} Specialized protocols \cite{BrickellPSW.2007, BarniFKLSS.2009, BostPTG.2015, WuFNL.2015, TaiMZC.2017, JoyeS18, TuenoKK19, KissNLAS19} exploit the domain knowledge of the problem at hand and make use of generic techniques only where it is necessary, resulting in more efficient solutions. Existing protocols for PDTE have several rounds requiring several interactions between the client and the server. Moreover, the communication cost depends on the size of the decision tree, while only a single classification is required by the client. Finally, they also require computational power from the client that depends on the size of the tree. 

\paragraph{Our Solution Approach} Our goal is to design and implement a novel client-server protocol that delegates the complete tree evaluation to the server while preserving privacy and keeping the performance acceptable. The idea is to use fully or somewhat homomorphic encryption (FHE/SHE) and evaluate the tree on ciphertexts encrypted under the client's public key. As a result, no intermediate or final computational result is revealed to the evaluating server.  However, since current somewhat homomorphic encryption schemes have high overhead, we combine efficient data representations with different algorithmic optimizations to keep the computational overhead and the communication cost low. At the end, the computational overhead might still be higher than in existing protocols, however the computation task can be parallelized resulting in a reduced computation time.
As a result, we are able to provide the first non-interactive protocol, that allows the client to delegate the evaluation to the server by sending an encrypted input and receiving only the encryption of the result. Finally, existing approaches are secure in the semi-honest model and can be made \emph{one-sided simulatable}\footnote{A 2-party protocol between parties $P_1$ and $P_2$ in which only $P_2$ receives an output, is \emph{one-sided simulatable} if it is private (via indistinguishability) against a corrupt $P_1$ and fully simulatable against a corrupt $P_2$ \cite{HazayL10}.} using techniques that may double the computation and communication costs. Our approach is one-sided simulatable by default, as the client does no more than encrypting its input and decrypting the final result of the computation (simulating the client is straightforward), while the server evaluates on ciphertexts encrypted with a semantically secure encryption under the client's public key.

\paragraph{Applications (ML-as-a-service)} Concrete motivation of our approach are machine learning settings (with applications in areas such as healthcare, finance etc.) where the server is computationally powerful, the client is computationally weak and the network connection is not very fast.
Many cloud providers are already proposing platforms that allow users to build machine learning applications \cite{bigML, MicrosoftAzureML, AmazonML, GoogleML, PredictionIOML}. A hospital may want to use such a platform to offer a medical expert system as a ML-as-a-service application to other doctors or even its patients. A software provider may leverage ML-as-a-service to allow its customers to detect the cause of a software error. Software systems use log files to collect information about the system behavior. In case of an error these log files can be used to find the cause of the crash. Both examples (medical data and log files) contain sensitive information which is worth protecting.

\paragraph{Contribution} Our contributions are as follows:
\begin{itemize}
	\item We propose a non-interactive protocol for PDTE. Our scheme allows the client to delegate the evaluation to the server by sending an encrypted input and receiving only the encryption of the result.
	\item We propose $\pdtebin$ which is an instantiation of the main protocol with binary representation of the input. Then we combine efficient data representations with different algorithmic optimizations to keep the computational overhead and the communication cost low.
	\item We propose $\pdteint$ which is an instantiation of the main protocol using arithmetic circuit, where the values are compared using  a modified variant of LinTzeng comparison protocol \cite{LinT05}.
	\item We provide correctness and security proofs of our scheme. Finally, we implement and benchmark both instantiations using HElib \cite{HaleviS14} and TFHE \cite{TFHE}. 
\end{itemize}

\paragraph{Structure} The remainder of the paper is structured as follows. We review related work in Section \ref{Related_work} and preliminaries in Section \ref{Preliminaries} before defining correctness and security of our protocol in Section \ref{Definitions}. The basic construction itself is described in Section \ref{Basic_Protocol}. In Section \ref{Binary_Implementation}, we describe implementation and optimization using a binary representation. In Section \ref{Arithmetic_Implementation}, we describe an implementation using an arithmetic circuit based on LinTzeng comparison protocol \cite{LinT05}. We discuss implementation and evaluation details in Section \ref{Evaluation} before concluding our work in Section \ref{Conclusion}. Due to space constraints, we discuss further details in the appendix.

\section{Related Work}
\label{Related_work}

Our work is related to secure multiparty computation (SMC) \cite{Yao.1982, Goldreich.2004, CramerDN2001, ChaumCD.1988, Ben-OrGW.1988, DamgardPSZ.2012, DamgardKLPSS.2013, KellerOS.2016}, private function evaluation (PFE) \cite{KolesnikovS_FC.2008, MohasselSS.2014} particularly privacy-preserving decision tree evaluation \cite{BrickellPSW.2007, BarniFKLSS.2009, BostPTG.2015, WuFNL.2015, TaiMZC.2017, JoyeS18, TuenoKK19, KissNLAS19} which we briefly review in this section and refer to the literature for more details.

Brikell et al.~\cite{BrickellPSW.2007} propose the first protocol for PDTE by combining homomorphic encryption (HE) and garbled circuits (GC) in a novel way. The server transforms the decision tree into a GC that is executed by the client. To allow the client to learn its garbling key, they combine homomorphic encryption and oblivious transfer (OT).
Although the evaluation time of Brikell et al.'s scheme is sublinear in the tree size, the secure program itself and hence the communication cost is linear and therefore not efficient for large trees. Barni et al.~\cite{BarniFKLSS.2009} improve the previous scheme by not including the leaf node in the transformed secure program, thereby reducing the computation costs by a constant factor.
Bost et al.~\cite{BostPTG.2015} represent the decision tree as a multivariate polynomial. The constants in the polynomial are the classification labels and the variables represent the results of the Boolean conditions at the decision nodes. The parties privately compute the values of the Boolean conditions by comparing each threshold with the corresponding attribute value encrypted under the client's public key. Finally, the server homomorphically evaluates the polynomial and returns the result to the client.  
Wu et al.~\cite{WuFNL.2015} use different techniques that require only additively HE (AHE). They also use the protocol from \cite{DamgardGK.2007} for comparison and reveal to the server comparison bits encrypted under the client's public key. The evaluation of the tree by the sever returns the index of the corresponding classification label to the client. Finally, an OT reveals the final result to the client.  
Tai et al.~\cite{TaiMZC.2017} use the comparison protocol of \cite{DamgardGK.2007} and AHE as well. They mark the left and right edge of each node with the cost $b$ and $1-b$ respectively, where $b$ is the result of the comparison at that node. Finally, they sum for each path of the tree the cost along it. The label of the path whose costs sum to zero, is the classification label.
%
%
Tueno et al.~\cite{TuenoKK19} represent the tree as an array. Then they execute $d$ – depth of the tree – comparisons, each performed using a small garbled circuit, which outputs secret-shares of the index of the next node in tree. They introduce a primitive called oblivious array indexing, that allow the parties to select the next node without learning it.
Kiss et al.~\cite{KissNLAS19} propose a modular design consisting of the sub-functionalities: selection of attributes, integer comparison, and evaluation of paths. They then explore the tradeoffs and performance of possible combinations of state-of-the-art protocols for privately computing the sub-functionalities.
De Cock et al.~\cite{CockDHKNNP.2016} follow the same idea as some previous schemes by first running comparisons. In contrast to all other protocols (ours included), which are secure in the computational setting, they operate in the information theoretic model using secret sharing based SMC and commodity-based cryptography \cite{Beaver.1997} to reduce the number of interactions.
Using a polynomial encoding of the inputs, Lu et al.~\cite{LuZS2018} propose a non-interactive comparison protocol called XCMP using BGV homomorphic scheme \cite{BrakerskiGV11}. They then implement the private decision tree protocol of Tai et al.~\cite{TaiMZC.2017} using XCMP which is \emph{output expressive} (i.e., it preserves additive homomorphism). The resulting decision tree protocol is non-interactive and efficient because of the small \emph{multiplicative depth}. However, it is not \emph{generic}, that is, it primarily works for small inputs and depends explicitly on BGV-type HE scheme. Moreover, it does not support \emph{SIMD} operations and is no longer output expressive as XCMP. Hence, it cannot be extended to a larger protocol (e.g., random forest \cite{Breiman2001}) while preserving the non-interactive property. Finally, its \emph{output length} (i.e., the number of resulted ciphertexts from server computation) is exponential in the depth of the tree, while the output length of our binary instantiation is at most linear in the depth of the tree and the integer instantiation can use SIMD to considerably reduce it. A comparison of decision protocols is summarized in Table \ref{Complexity_of_PDTEP} and \ref{Feature_of_1Round_PDTEP}. A more detailed complexity analysis is described in Appendix \ref{complexity_analysis}.

\begin{table}
	\centering
		\begin{tabular}{|c|l|l|}
		\hline
		\multicolumn{1}{|c|}{\textnormal{Symbol}}  & {\textnormal{Interpretation}}	\\  \hline \hline
		\multicolumn{1}{|c|}{\textnormal{$\inputlen$}}  & {\textnormal{Bit length of attribute values}} \\ \hline
		\multicolumn{1}{|c|}{\textnormal{$n$}}  & {\textnormal{Dimension of the attribute vector}} \\ \hline
		\multicolumn{1}{|c|}{\textnormal{$x = x_0, \ldots, x_{n-1}$}}  & {\textnormal{Attribute vector}} \\ \hline 
		\multicolumn{1}{|c|}{$|\alpha|$}  & {Bitlength of integer $\alpha$, e.g., $|x_i| = \inputlen$} \\ \hline
		\multicolumn{1}{|c|}{$\bitrep{x_i} = x_{i\inputlen} \ldots x_{i1}$}  & {Bit representation of $x_i$ with most significant bit $x_{i\inputlen}$ } \\ \hline
		
		\multicolumn{1}{|c|}{\textnormal{$M$}}  & {\textnormal{Number of nodes}} \\ \hline
		\multicolumn{1}{|c|}{\textnormal{$m$}}  & {\textnormal{Number of decision nodes}} \\ \hline
		\multicolumn{1}{|c|}{\textnormal{$d$}}  & {\textnormal{Depth of the decision tree}} \\ \hline 

		\multicolumn{1}{|c|}{\textnormal{$\ctxtrep{\alpha}$}}  & {\textnormal{Ciphertext HE of a plaintext $\alpha$}} \\ \hline
		\multicolumn{1}{|c|}{$\lsem \bitrep{x_i} \rsem$}  & {Bitwise encryption $(\lsem x_{i\inputlen}\rsem, \ldots, \lsem x_{i1} \rsem)$ of $x_i$} \\ \hline
		\multicolumn{1}{|c|}{$\lsem \alpha_1 | \ldots | \alpha_{\slotnb} \rsem$}  & {Packed ciphertext containing plaintexts $\alpha_1, \ldots, \alpha_{\slotnb}$} \\ \hline
		\multicolumn{1}{|c|}{$\slotnb$}  & {Number of slots in a packed ciphertext} \\ \hline
		\multicolumn{1}{|c|}{$\lsem \vec{x}_i \rsem$}  & {Packed ciphertext $\lsem x_{i\inputlen} | \ldots | x_{i1} | 0 | \ldots | 0 \rsem$ of $\bitrep{x_i}$} \\ \hline
		\end{tabular}
		\captionsetup{justification=centering}
		\caption[Notations]{Notations.}
		\label{Notation_Table}
\end{table}

\begin{table}[t]
\begin{center}
  \begin{tabular}{| c | l | c | c | c | c | c | c |}
    \hline
    \multicolumn{1}{|c|}{\textnormal{Scheme}}                  & {\textnormal{Rounds}}  & {\textnormal{Tools}}      & {\textnormal{Commu-}}     & {\textnormal{Compa-}} & {\textnormal{Leakage}}  \\ 
		\multicolumn{1}{|c|}{\textnormal{ }}                       &            &            & {\textnormal{nication}}     &  {\textnormal{risons}} &  \\ \hline \hline
    \multicolumn{1}{|c|}{\textnormal{\cite{BrickellPSW.2007}}} & {\textnormal{$\approx$5}} & {\textnormal{HE+GC}}      & {\textnormal{$\bigO{2^d}$}}            & {\textnormal{$d$}} & $m, d$       \\ \hline
    \multicolumn{1}{|c|}{\textnormal{\cite{BarniFKLSS.2009}}}  & {\textnormal{$\approx$4}} & {\textnormal{HE+GC}}      & {\textnormal{$\bigO{2^d}$}}            & {\textnormal{$d$}} & $m, d$       \\ \hline
    \multicolumn{1}{|c|}{\textnormal{\cite{BostPTG.2015}}}     & {\textnormal{$\geq$6}}    & {\textnormal{FHE/SHE}}    & {\textnormal{$\bigO{2^d}$}}            & {\textnormal{$m$}} & $m$       \\ \hline
    \multicolumn{1}{|c|}{\textnormal{\cite{WuFNL.2015}}}       & {\textnormal{6}}          & {\textnormal{HE+OT}}      & {\textnormal{$\bigO{2^d}$}}            & {\textnormal{$m$}} & $m$       \\ \hline
		\multicolumn{1}{|c|}{\textnormal{\cite{TaiMZC.2017}}}      & {\textnormal{4}}          & {\textnormal{HE}}         & {\textnormal{$\bigO{2^d}$}}            & {\textnormal{$m$}} & $m$      \\ \hline
    \multicolumn{1}{|c|}{\textnormal{\cite{CockDHKNNP.2016}}}  & {\textnormal{$\approx$9}} & {\textnormal{SS}}         & {\textnormal{$\bigO{2^d}$}}            & {\textnormal{$m$}} & $m, d$       \\ \hline
    \multicolumn{1}{|c|}{\textnormal{\cite{TuenoKK19}}}                    & {\textnormal{$\bigO{d}$}}   & {\textnormal{GC,OT}}      & {\textnormal{$\bigO{2^d}$}}     & {\textnormal{$d$}} & $m, d$       \\ 
		\multicolumn{1}{|c|}{\textnormal{ }}      &  {\textnormal{}} & {\textnormal{ORAM}} & {\textnormal{$\bigO{d^2}$}}     &   &      \\  
	\hline
	\multicolumn{1}{|c|}{\textnormal{\cite{LuZS2018}}}                    & {\textnormal{$1$}}   & {\textnormal{FHE/SHE}}      & {\textnormal{$\bigO{2^d}$}}     & {\textnormal{$m$}}  & $m$      \\
	\hline
	\multicolumn{1}{|c|}{\textnormal{\pdtebin}}                    & {\textnormal{$1$}}   & {\textnormal{FHE/SHE}}      & {\textnormal{$\bigO{1}$ or $\bigO{d}$}}     & {\textnormal{$m$}}  & -      \\ 
		\multicolumn{1}{|c|}{\textnormal{\pdteint}}      &  {\textnormal{$1$}} & {\textnormal{}} & {\textnormal{$\bigO{2^d/s}$}}     &    & $m$     \\ 
		\hline
  \end{tabular}
	\captionsetup{justification=centering}
  \caption[Comparison of Private Decision Tree Protocols.]{Comparison of PDTE protocols.}
  \label{Complexity_of_PDTEP}
\end{center}
\end{table}

\begin{table}[t]
\begin{center}
  \begin{tabular}{|c| l | c | c | c | c | c | c |}
    \hline
    \multicolumn{1}{|c|}{\textnormal{Scheme}}                  & {\textnormal{SIMD}}  & {\textnormal{Generic}}      & {\textnormal{Output-}}     & {\textnormal{Multiplicative}} & {\textnormal{Output}}  \\ 
		\multicolumn{1}{|c|}{\textnormal{ }}                       &            &            & {\textnormal{expressive}}     &  {\textnormal{Depth}} & {\textnormal{Length}}  \\ \hline \hline
	\multicolumn{1}{|c|}{\textnormal{\cite{LuZS2018}}}                    & {\textnormal{no}}   & {\textnormal{no}}      & {\textnormal{no}}     & {\textnormal{$3$}} &  {\textnormal{$2^{d+1}$}}      \\
	\hline
	\multicolumn{1}{|c|}{\textnormal{\pdtebin}}                    & {\textnormal{yes}}   & {\textnormal{yes}}      & {\textnormal{yes}}     & {\textnormal{$|\inputlen| + |d| + 2$}}  &  {\textnormal{$1$ or $d$}}      \\ 
		\multicolumn{1}{|c|}{\textnormal{\pdteint}}      &  {\textnormal{yes}} & {\textnormal{yes}} & {\textnormal{no}}     &   {\textnormal{$|\inputlen| + 1$}}   &  {\textnormal{$\ceil{2^d/s}$}}   \\ 
		\hline
  \end{tabular}
	\captionsetup{justification=centering}
  \caption[Comparison of 1-round Private Decision Tree Protocols]{Comparison of 1-round PDTE protocols.}
  \label{Feature_of_1Round_PDTEP}
\end{center}
\end{table}

\section{Preliminaries}
\label{Preliminaries} 
In this section, we present the background concepts for the remainder of the paper. The core concept is fully/somewhat homomorphic encryption. For ease of exposition and understanding, we abstract away the mathematical technicalities behind homomorphic encryption and refer the reader to the relevant literature \cite{HEStandard2018, BrakerskiGV11, Gentry.2009, SmartV2014, ChillottiGGI16, ChillottiGGI17, TFHE, ChillottiGGI18, cuFHE, dai2015cuhe, sealcrypto}. 

\paragraph{Homomorphic Encryption}

A homomorphic encryption (HE) allows computations on ciphertexts by generating an encrypted result whose decryption matches the result of a function on the  plaintexts. In this paper, we focus on homomorphic encryption schemes (particularly lattice-based) that allow many chained additions and multiplications to be computed on plaintext homomorphically. In these schemes, the plaintext space is usually a ring $\mathbb{Z}_q[X]/(X^{\poldegree}+1)$, where $q$ is prime and $\poldegree$ might be a power of 2. A HE scheme consists of the following algorithms:

\begin{itemize}
\item $\pk, \sk, \ek \leftarrow \kgen(\secprm)$: This probabilistic algorithm takes a security parameter $\secprm$ and outputs 
public, private and evaluation key $\pk$, $\sk$ and $\ek$.
\item $\ctxt \leftarrow \enc(\pk, \ptxt)$: This probabilistic algorithm takes $\pk$ and a message $\ptxt$ and outputs a ciphertext $\ctxt$. We will use $ \lsem \ptxt \rsem$ as a shorthand notation for $\enc(\pk, \ptxt)$.
\item $\ctxt \leftarrow \eval(\ek, f, \ctxt_1, \ldots, \ctxt_n)$: This probabilistic algorithm takes $\ek$, an $n$-ary function $f$ and $n$ ciphertexts $\ctxt_1, \ldots \ctxt_n$ and outputs a ciphertext $\ctxt$.
\item $\ptxt' \leftarrow \dec(sk, \ctxt)$: This deterministic algorithm takes $\sk$ and a ciphertext $\ctxt$ and outputs a message $\ptxt'$.
\end{itemize}

\noindent We require IND-CPA and the following correctness conditions $\forall \ptxt_1, \ldots, \ptxt_n$ and $\vec{\ptxt} = \ptxt_1, \ldots, \ptxt_n$:
\begin{itemize}
	\item $\dec(\sk, \enc(\pk, \ptxt_i)) = \dec(\sk, \lsem \ptxt_i \rsem) = \ptxt_i,$
	\item $\dec(\sk, \eval(\ek, f, \lsem \ptxt_1 \rsem, \ldots, \lsem \ptxt_n \rsem)) = \\ \dec(\sk, \lsem f(\vec{\ptxt}) \rsem)$.
\end{itemize}

The encryption algorithm $\enc$ adds \enquote{noise} to the ciphertext which increases during homomorphic evaluation. While addition of ciphertexts increases the noise linearly, the multiplication increases it exponentially \cite{BrakerskiGV11}. If the noise become too large then correct decryption is no longer possible. To prevent this from happening, one can either keep the circuit's depth of the function $f$ low enough or use the \emph{refresh} algorithm. This algorithm consists either of a \emph{bootstrapping} procedure, which takes a ciphertext with large noise and outputs a ciphertext of the same message with a fixed amount of noise; or a \emph{key-switching procedure}, which takes a ciphertext under one key and outputs a ciphertext of the same message under a different key \cite{HEStandard2018}. In this paper, we will consider both bootstrapping and the possibility of keeping the circuit's depth low 
by designing our PDTE using so-called leveled fully homomorphic encryption.
A leveled fully homomorphic encryption (FHE) has an extra parameter $\nblevel$ such that the scheme can evaluate all circuits of depth at most $\nblevel$ without bootstrapping.

\paragraph{Homomorphic Operations}
We assume a BGV type homomorphic encryption scheme \cite{BrakerskiGV11}. Plaintexts can be encrypted using an integer representation (an integer $x_i$ is encrypted as $\ctxtrep{x_i}$) or a binary representation (each bit of the bit representation  $\bitrep{x_{i}} = x_{i\inputlen}\ldots x_{i1}$ is encrypted). 
We describe below homomorphic operations in the binary representation (i.e., arithmetic operations $\bmod~2$). They work similarly in the integer representation.

The FHE scheme might support Smart and Vercauteren's ciphertext packing (SVCP) technique \cite{SmartV2014} to pack many plaintexts in one ciphertext.
Using SVCP, a ciphertext consists of a fixed number $\slotnb$ of slots, each capable of holding one plaintext, i.e. $\lsem \cdot | \cdot | \ldots | \cdot \rsem$. The encryption of a bit $b$ replicates $b$ to all slots, i.e., $\lsem b \rsem = \lsem b | b | \ldots | b \rsem$. However, we can also pack the bits of $\bitrep{x_{i}}$ in one ciphertext and will denote it by $\lsem \vec{x}_i \rsem = \lsem x_{i\inputlen} | \ldots | x_{i1}|0|\ldots|0 \rsem$. 

The computation relies on some built-in routines, that allow homomorphic operations on encrypted data. The relevant routines for our scheme are: addition ($\sheadd$), multiplication ($\shemul$) and comparison ($\shecmp$). These routines are compatible with the ciphertext packing technique (i.e., operations are replicated on all slots in a SIMD manner).

The routine $\sheadd$ takes two or more ciphertexts and performs a component-wise addition modulo two, i.e., we have:
$$\sheadd( \lsem b_{i1} | \ldots | b_{i\slotnb} \rsem,  \lsem b_{j1} | \ldots | b_{j\slotnb} \rsem) =  \lsem b_{i1} \oplus b_{j1} | \ldots | b_{i\slotnb} \oplus b_{j\slotnb} \rsem.$$ 
Similarly, $\shemul$ performs component-wise multiplication modulo two, i.e., we have:
$$\shemul( \lsem b_{i1} | \ldots | b_{i\slotnb} \rsem,  \lsem b_{j1} | \ldots | b_{j\slotnb} \rsem) =  \lsem b_{i1} \cdot b_{j1} | \ldots | b_{i\slotnb} \cdot b_{j\slotnb} \rsem.$$ 
We will also denote addition and multiplication by $\oplus$ and $\odot$, respectively.

Let $x_i, x_j$ be two integers, $b_{ij} = [x_i > x_j]$ and $b_{ji} = [x_j > x_i]$, the routine $\shecmp$ takes  $\lsem \bitrep{x_{i}} \rsem,  \lsem \bitrep{x_{j}} \rsem$, compares $x_i$ and $x_j$ and returns $\lsem b_{ij} \rsem$, $\lsem b_{ji} \rsem$:
$$(\lsem b_{ij} \rsem, \lsem b_{ji} \rsem) \gets \shecmp(\lsem \bitrep{x_{i}} \rsem,  \lsem \bitrep{x_{j}} \rsem).$$
Note that, if the inputs to $\shecmp$ encrypt the same value, then the routine outputs two ciphertexts of 0. This routine implements the comparison circuit described in \cite{CheonKK15, CheonKK16, CheonKL15}.

If ciphertext packing is enabled, then
we also assume the HE supports shift operations. Given a packed ciphertext $\ctxtrep{b_1 | \ldots |b_{\slotnb}}$, the \emph{shift left} operation 
shifts all slots to the left by a given offset, using zero-fill, i.e., shifting $\ctxtrep{b_1 | \ldots |b_{\slotnb}}$ by $i$ positions returns $\ctxtrep{b_i | \ldots |b_{\slotnb} | 0 | \ldots | 0}$.
The \emph{shift right} operation is defined similarly for shifting to the right.

\section{Definitions}
\label{Definitions}

In this section, we introduce relevant definitions and notations for our scheme. Our definitions and notations 
are similar to previous work \cite{WuFNL.2015, CockDHKNNP.2016, TaiMZC.2017, TuenoKK19}. With $[a,b]$, we denote the set of all integers from $a$ to $b$. Let $c_0, \ldots, c_{k-1}$ be the classification labels, $k\in\mathbb{N}_{>0}$.

\begin{definition}[Decision Tree]
\label{Decision_Tree_Def}
A \emph{decision tree} (DT) is a function 
$\mathcal{T}:\mathbb{Z}^n \rightarrow \{c_0, \ldots, c_{k-1}\}$ 
that maps an 
\emph{attribute vector} $x=(x_0, \ldots, x_{n-1})$ to a finite set of \emph{classification labels}. The tree consists of: 
\begin{itemize}
	\item \emph{internal} or \emph{decision nodes} containing a test condition
	\item \emph{leave nodes} containing a classification label.
\end{itemize}
A \emph{decision tree model} consists of a decision tree and the following functions:
\begin{itemize}
	\item a function ~$\tfun$~ that assigns to each decision node a \emph{threshold} value, 
	$\tfun: [0, m-1] \rightarrow \mathbb{Z},$
	\item a function $\afun$ that assigns to each decision node an \emph{attribute index},  
	$\afun: [0, m-1] \rightarrow [0, n-1], ~\mbox{and}$
	\item a labeling function $\cfun$ that assigns to each leaf node a label,  
	$\cfun: [m, M-1] \rightarrow \{c_0, \ldots, c_{k-1}\}.$
\end{itemize}
The decision at each decision node is a \enquote{greater-than} comparison between the assigned threshold and attribute values, i.e., the decision at node $v$ is $[x_{\afun(v)} \geq \tfun(v)]$.
\end{definition}

\begin{definition}[Node Indices]
\label{Node_Indices_Def}
Given a decision tree, the \emph{index} of a node is its order as computed by \emph{breadth-first search (BFS)} traversal, starting at the root with index 0. If the tree is complete, then a node with index $v$ has left child $2v+1$ and right child $2v+2$.
\end{definition}

We will also refer to the node with index $v$ as the node $v$. W.l.o.g, we will use $[0, k-1]$ as classification labels (i.e., $c_j = j$ for $0\leq j \leq k-1$) and we will label the first (second, third, $\ldots$) leaf in BFS traversal with classification label 0 (1, 2, $\ldots$). For a complete decision tree with depth $d$, the leaves have indices ranging from $2^d, 2^d+1, \ldots 2^{d+1}-2$ and classification labels ranging from $0, \ldots, 2^d-1$ respectively. Since the classification labeling is now independent of the tree, we use $\mathcal{M} = (\mathcal{T}, \tfun, \afun)$ to denote a \emph{decision tree model} consisting of a tree $~\mathcal{T}$ and the labeling functions $\tfun, \afun$ as defined above. We also assume that the tree parameters $d, m, M$ can be derived from $~\mathcal{T}$.

\begin{definition}[Decision Tree Evaluation]
\label{Decision_Tree_Evaluation_Def}
Given $x=(x_0, \ldots, x_{n-1})$ and $\mathcal{M} = (\mathcal{T}, \tfun, \afun)$, then starting at the root, the \emph{Decision Tree Evaluation} (DTE) evaluates at each reached node $v$ the decision $b \leftarrow [x_{\afun(v)} \geq \tfun(v)]$ and moves either to the left (if $b = 0$) or right (if $b=1$) subsequent node. The evaluation returns the label of the reached leaf as result of the computation. We denote this by $\mathcal{T}(x)$.
\end{definition}

\begin{definition}[Private DTE]
\label{Private_Decision_Tree_Evaluation_Def}
Given a \emph{client} with a \emph{private} $x=(x_0, \ldots, x_{n-1})$ and a \emph{server} with a \emph{private} $\mathcal{M} = (\mathcal{T}, \tfun, \afun)$, a \emph{private DTE (PDTE)} functionality evaluates the model $\mathcal{M}$ on input $x$, then reveals to the client the \emph{classification label} $\mathcal{T}(x)$ and nothing else, while the server learns \emph{nothing}, i.e., 
$$\func{F}{PDTE}(\mathcal{M}, x) \rightarrow (\emptystring, \mathcal{T}(x)).$$
\end{definition}

\begin{definition}[Correctness]
\label{Correctness_Def}
Given a \emph{client} with a \emph{private} $x=(x_0, \ldots, x_{n-1})$ and a \emph{server} with a \emph{private} $\mathcal{M} = (\mathcal{T}, \tfun, \afun)$, a protocol $\Pi$ \emph{correctly implements} a PDTE functionality if after the computation it holds for the result $c$ obtained by the client that $c = \mathcal{T}(x)$.
\end{definition}
Besides correctness, parties must learn only what they are allowed to. To formalize this, we need the following two definitions.
A function $\mu: \mathbb{N} \rightarrow \mathbb{R}$ is \emph{negligible}, if for every positive polynomial $p(.)$ there exists an $\epsilon$ such that for all $n > \epsilon$: $\mu(n)<1/p(n)$. 
Two distributions $\mathcal{D}_1$ and $\mathcal{D}_2$ are \emph{computationally indistinguishable} (denoted $\mathcal{D}_1 \stackrel{c}{\equiv} \mathcal{D}_1$) if no probabilistic polynomial time (PPT) algorithm can distinguish them except with negligible probability. 

In SMC protocols, the \emph{view} of a party consists of its input and the sequence of messages that it has received during the protocol execution \cite{Goldreich.2004}. The protocol is said to be secure, if for each party, one can construct a simulator that, given only the input of that party and the output, can generate a distribution that is computationally indistinguishable to the party's view. 
\begin{definition}[PDTE Security]
\label{Semi-Honest_Security_Def}
Given a \emph{client} with a \emph{private} $x=(x_0, \ldots, x_{n-1})$ and a \emph{server} with a \emph{private} $\mathcal{M} = (\mathcal{T}, \tfun, \afun)$, a protocol $\prot{PDTE}$ \emph{securely implements} the PDTE functionality in the \emph{semi-honest model} if the following holds:
\begin{itemize}
	\item there exists a PPT algorithm $\simul{S}{PDTE}$ that simulates the server's view $\view{S}{\prot{PDTE}}$ given only the private decision tree model $(\mathcal{T}, \tfun, \afun)$ such that:
	\begin{eqnarray}
		\simul{S}{PDTE}(\mathcal{M}, \emptystring) & \stackrel{c}{\equiv} & \view{S}{\prot{PDTE}}(\mathcal{M}, x), \label{eq:PDTE_Security_Eqn_Server}
\end{eqnarray}
	\item there exists a PPT algorithm $\simul{C}{PDTE}$ that simulates the client's view $\view{C}{\prot{PDTE}}$ given only the depth $d$ of the tree, $x=(x_0, \ldots, x_{n-1})$ and a classification label $\mathcal{T}(x) \in \{0, \ldots, k-1\}$ such that:
	\begin{eqnarray}
		\simul{C}{PDTE}(\left\langle d, x\right\rangle, \mathcal{T}(x)) & \stackrel{c}{\equiv} & \view{C}{\prot{PDTE}}(\mathcal{M}, x). \label{eq:PDTE_Security_Eqn_Client}
\end{eqnarray}
\end{itemize}
A protocol $\prot{PDTE}$ \emph{securely implements} the PDTE functionality with \emph{one-sided simulation} if the following conditions hold:
\begin{itemize}
	\item for every pair $x, x'$ of different client's inputs, it holds:
	\begin{eqnarray}
		\view{S}{\prot{PDTE}}(\mathcal{M}, x) & \stackrel{c}{\equiv} & \view{S}{\prot{PDTE}}(\mathcal{M}, x'), \label{eq:PDTE_Security_OSS_Eqn_Server}
\end{eqnarray}
	\item  $\prot{PDTE}$ is simulatable against every PPT adversary controlling $C$.
\end{itemize}
\end{definition}
Note that for the one-sided simulation, the requirement in Equation \ref{eq:PDTE_Security_OSS_Eqn_Server} is that the protocol should be indistinguishable against any PPT adversary that controls the server. This means, the server should not be able to distinguish between the case where the client uses $x$ and the case where it uses $x'$. Moreover, the protocol should be simulatable against any adversary controlling the client \cite{HazayL10}. 
\section{The Basic Protocol}
\label{Basic_Protocol}
In this section, we present a modular description of our basic protocol. We start by describing the data structure.

\subsection{Data Structure}
We follow the idea of some previous protocols \cite{BostPTG.2015,TaiMZC.2017, CockDHKNNP.2016} of marking edges of the tree with comparison result. So if the comparison at node $v$ is the bit $b$ then we mark the right edge outgoing from $v$ with $b$ and the left edge with $1-b$. For convenience, we will instead store this information at the child nodes of $v$ and refer to it as $\cmp$.
\begin{definition}[Data Structure]
For a decision tree model $\mathcal{M} = (\mathcal{T}, \tfun, \afun)$, we let $\node$ be a data structure that for each node $v$ defines the following fields:
\begin{itemize}
	\item $v.\thr$ stores the threshold $\tfun(v)$ of the node $v$
	\item $v.\aindex$ stores the associated index $\afun(v)$
	\item $v.\pnode$ stores the pointer to the parent node which is null for the root node
	\item $v.\lnode$ stores the pointer to the left child node which is null for each leaf node
	\item $v.\rnode$ stores the pointer to the right child node which is null for each leaf node
	\item $v.\cmp$ is computed during the tree evaluation and stores the comparison bit 
	$$b \leftarrow [x_{\afun(v.\pnode)} \geq \tfun(v.\pnode)]$$ 
	if $v$ is a right node. Otherwise it stores $1-b$.
	\item $v.\clabel$ stores the classification label if $v$ is a leaf node and the empty string otherwise.
\end{itemize}
We use $\nodeset$ to denote the set of all decision nodes and $\leafset$ the set of all leave nodes of $\mathcal{M}$. As a result, we use the equivalent notation  $\mathcal{M} = (\mathcal{T}, \tfun, \afun) = (\nodeset, \leafset)$.
\end{definition}
With the data structure defined above, we now define the classification function as follows.

\begin{definition}[Classification Function]
Let the attribute vector be $x=(x_0, \ldots, x_{n-1})$ and the decision tree model be $\mathcal{M} = (\nodeset, \leafset)$. We define the \emph{classification function} to be 
$$\classfun(x, \mathcal{M}) = \travfun(x, \rootnode),$$
where $\rootnode$ is the root node and $\travfun$ is the \emph{traverse function} define as:
$$
\travfun(x, v) =
\begin{cases}
	\travfun(x, v.\lnode) & \mbox{if} ~ v \in \nodeset ~ \mbox{and} ~ x_{v.\aindex} < v.\thr \\
	\travfun(x, v.\rnode) & \mbox{if} ~ v \in \nodeset ~ \mbox{and} ~ x_{v.\aindex} \geq v.\thr \\
	v                     & \mbox{if} ~ v \in \leafset 
\end{cases}
$$
\end{definition}

\begin{lemma}
\label{thr:class_function_corretness}
Let $x=(x_0, \ldots, x_{n-1})$ be an attribute vector and $\mathcal{M} = (\mathcal{T}, \tfun, \afun)  = (\nodeset, \leafset)$ a decision model. We have 
$$\mathcal{T}(x) = b \cdot \travfun(x, \rootnode.\rnode) + (1-b) \cdot \travfun(x, \rootnode.\lnode),$$
where $b = [x_{\afun(\rootnode)} \geq \tfun(\rootnode)]$ is the comparison at the root node.
\end{lemma}
\begin{proof}
The proof follows by induction on the depth of the tree. In the base case, we have a tree of depth one (i.e., the root and two leaves). In the induction step, we have two trees of depth $d$ and we joint them by adding a new root. 
\end{proof}

\subsection{Algorithms}

\paragraph{Initialization}
The Initialization consists of a one-time key generation. The client generates appropriate triple $(\pk, \sk, \ek)$ of public, private and evaluation keys for a homomorphic encryption scheme. Then the client sends $(\pk, \ek)$ to the server.  For each input classification, the client just encrypts its input and sends it to the server. To reduce the communication cost of sending client's input, one can use a trusted randomizer that does not take part in the real protocol and is not allowed to collaborate with the server. The trusted randomizer generates a list of random strings $r$ and sends the encrypted strings $\ctxtrep{r}$ to server and the list of $r$'s to the client. For an input $x$, the client then sends $x+r$ to the server in the real protocol. This technique is similar to the commodity based cryptography \cite{Beaver.1997} with the difference that the client can play the role of the randomizer itself and sends the list of $\ctxtrep{r}$'s (when the network is not too busy) before the protocol's start.

\paragraph{Computing Decision Bits}
The server starts by computing for each node $v \in \mathcal{D}$ the comparison bit $b \leftarrow [x_{\afun(v)} \geq \tfun(v)]$ and stores $b$ at the right child node ($v.\rnode.\cmp = b$) and $1-b$ at the left child node ($v.\lnode.\cmp = 1-b$). It is illustrated in Algorithm \ref{Eval_Nodes_Algo}.

\begin{figure}[tbp]
	\renewcommand{\figurename}{Algorithm}
	\begin{mdframed}
		\begin{algorithmic}[1]
		\Function {\evalnode}{$\nodeset, \ctxtrep{x}$}
			\ForEach {$v \in \nodeset$}
				\State $\ctxtrep{b} \gets \ctxtrep{[x_{v.\aindex} \geq v.\thr]} $
				\State $\ctxtrep{v.\rnode.\cmp} \gets \ctxtrep{b} $
				\State $\ctxtrep{v.\lnode.\cmp} \gets \ctxtrep{1-b} $
			\EndFor
		\EndFunction
		\end{algorithmic}
	\end{mdframed}
	\caption[Algorithm for Computing a Decision Bit]{Computing a Decision Bit}
	\label{Eval_Nodes_Algo}
\end{figure}

\paragraph{Aggregating Decision Bits}
Then for each leaf node $v$, the server aggregates the comparison bits along the path from the root to $v$. We implement it using a queue and traversing the tree in BFS as illustrated in Algorithm \ref{Eval_Paths_Algo}.

\begin{figure}[tbp]
	\renewcommand{\figurename}{Algorithm}
	\begin{mdframed}
		\begin{algorithmic}[1]
		\Function {\evalpath}{$\nodeset, \leafset$}
			\State let $Q$ be a queue
			\State $Q.\enQ(\rootnode)$
			\While  { $Q.\emptyQ() = \false $} 
				\State $v \gets Q.\deQ()$
				\State $\ctxtrep{v.\lnode.\cmp} \gets \ctxtrep{v.\lnode.\cmp} \odot \ctxtrep{v.\cmp}$, 
				\State $\ctxtrep{v.\rnode.\cmp} \gets \ctxtrep{v.\rnode.\cmp} \odot \ctxtrep{v.\cmp}$
				\If {$v.\lnode \in \nodeset$} 
					\State $Q.\enQ(v.\lnode)$
				\EndIf
				\If {$v.\rnode \in \nodeset$} 
					\State $Q.\enQ(v.\rnode)$
				\EndIf
			\EndWhile
		\EndFunction
		\end{algorithmic}
	\end{mdframed}
	\caption[Algorithm for Aggregating Decision Bits]{Aggregating Decision Bits}
	\label{Eval_Paths_Algo}
\end{figure}

\paragraph{Finalizing}
After Aggregating the decision bits along the path to the leave nodes, each leaf node $v$ stores either $v.\cmp = 0$ or $v.\cmp = 1$. Then, the server aggregates the decision bits at the leaves by computing for each leaf $v$ the value $\ctxtrep{v.\cmp} \odot \ctxtrep{v.\clabel}$ and summing all the results. This is illustrated in Algorithm \ref{Finalizing_Algo}.

\begin{figure}[tbp]
	\renewcommand{\figurename}{Algorithm}
	\begin{mdframed}
		\begin{algorithmic}[1]
		\Function {\evalleaves}{$\leafset$}
			\State $\ctxtrep{\result} \gets \ctxtrep{0}$
			\ForEach {$v \in \leafset$}
				\State $\ctxtrep{\result} \gets \ctxtrep{\result} \oplus (\ctxtrep{v.\cmp} \odot \ctxtrep{v.\clabel})$
			\EndFor
			\State \Return $\ctxtrep{\result}$
		\EndFunction
		\end{algorithmic}
	\end{mdframed}
	\caption[Algorithm for Aggregating Paths Results]{Finalizing}
	\label{Finalizing_Algo}
\end{figure}

\paragraph{Putting It All Together}
As illustrated in Algorithm \ref{Basic_PDTE_Algo}, the whole computation is performed by the server. It sequentially computes the algorithms described above and sends the resulting ciphertext to the client. The client decrypts and outputs the resulting classification label. The correctness is straightforward and follows from Lemma \ref{thr:class_function_corretness}. The algorithms are straightforward and easy to understand. However, their naive application is inefficient. 

\begin{figure}[tbp]
\renewcommand{\figurename}{Protocol}
\begin{center}
\fbox{
\pseudocode{%
\\
\textbf{ Client} \< \< \textbf{ Server}  \\[0.1\baselineskip][\hline]
 \<\< \\[-0.5\baselineskip]
\textbf{Input:}~ x  \<\< \textbf{Input:}~ \mathcal{M} = (\nodeset, \leafset)  \\
\textbf{Output:}~ \mathcal{T}(x) \<\< \textbf{Output:}~ \emptystring \\ [0.1\baselineskip][\hline]
\<\< \\[-0.5\baselineskip]
 \< \sendmessageright*[2cm]{\ctxtrep{x}} \< \\
  \<\< \evalnode(\nodeset, \ctxtrep{x}) \\
  \<\< \evalpath(\nodeset, \leafset) \\
  \<\< \ctxtrep{\mathcal{T}(x)} \gets \evalleaves(\leafset) \\
 \< \sendmessageleft*[2cm]{\ctxtrep{\mathcal{T}(x)}} \< \\
 }
}
\end{center}
\caption[Overview of Our 1-round PDTE]{The Basic Protocol}
\label{Basic_PDTE_Algo}
\end{figure}

\section{Binary Implementation}
\label{Binary_Implementation}
In this section, we describe $\pdtebin$, an instantiation of the basic scheme that requires encoding the plaintexts using their bit representation. Hence, ciphertexts encrypt bits and arithmetic operations are done $\bmod~2$.
\subsection{Input Encoding}
In this implementation, we encrypt plaintext bitwise. For each plaintext $x_i$ with bit representation $\bitrep{x_i} = x_{i\inputlen} \ldots x_{i1}$, we use $\ctxtrep{\bitrep{x_i}}$ to denote the vector $(\ctxtrep{x_{i\inputlen}}, \ldots, \ctxtrep{x_{i1}})$, consisting of encryptions of the bits of $x_i$. As a result, the client needs to send $n\inputlen$ ciphertexts for the $n$ attribute values. Unfortunately, homomorphic ciphertexts might be quite large. We can already use the trusted randomizer as explained before to send blinded inputs instead of ciphertexts in this phase. This, however, improves only the online communication. We additionally want to use the SVCP SIMD technique that allows to pack many plaintexts into the same ciphertext and manipulate them together during homomorphic operations.   

\subsection{Ciphertext Packing}
In the binary encoding, ciphertext packing means that each ciphertext encrypts $\slotnb$ bits, where $\slotnb$ is the number of slots in the ciphertext. Then we can use this property in three different ways. First, one could pack the bit representation of each classification label in a single ciphertext and allow the server to send back a single ciphertext to the client. Second, one could encrypt several attributes together and classify them with a single protocol evaluation. Finally, one could encrypt multiple decision node thresholds that must be compared to the same attribute in the decision tree model.
%


\paragraph{Packing Classification Label's Bits} Aggregating the decision bits using Algorithm \ref{Eval_Paths_Algo} produces for each leaf $v \in \leafset$ a decision bit $\ctxtrep{b_v}$ which encrypts 1 for the classification leaf and 0 otherwise. Moreover, because of SVCP, the bit $b_{v}$ is replicated to all slots. Now, let $k$ be the number of classification labels (i.e., $|\leafset| = k$) and its bitlength be $|k|$. For each $v \in \leafset$, we let $c_{v}$ denote the classification label $v.\clabel$ which is $|k|$-bit long and has bit representation $\bitrep{c_{v}} = c_{v|k|}\ldots c_{v1}$ with corresponding packed encryption  $\ctxtrep{\vec{c_{v}}} = \ctxtrep{c_{v|k|} | \ldots | c_{v1} | 0 | \ldots | 0}$. As a result, computing $\ctxtrep{b_v} \odot \ctxtrep{\vec{c_{v}}}$ for each leaf $v \in \leafset$ and summing over all leaves results in the correct classification label. Note that, this assumes that one is classifying only one vector and not many as in the the next case.

\paragraph{Packing Attribute Values} Let $x^{(1)}, \ldots, x^{(\slotnb)}$ be $\slotnb$ possible attribute vectors with $x^{(\slotct)} = [x_{1}^{(\slotct)}, \ldots, x_{n}^{(\slotct)}]$, $1 \leq \slotct \leq \slotnb$. For each $x_{i}^{(\slotct)}$, let $\bitrepx{x_i}{\slotct} = x_{i\inputlen}^{(\slotct)}, \ldots, x_{i1}^{(\slotct)}$ be the bit representation. Then, the client generates for each attribute $x_i$ the ciphertexts $\ctxtrep{cx_{i\inputlen}}, \ldots, \ctxtrep{cx_{i2}}, \ctxtrep{cx_{i1}}$ as illustrated in Equation \ref{eqn:CP_attribute_packing}.

\begin{equation}
\begin{aligned}
\ctxtrep{cx_{i1}} &= \lsem x_{i1}^{(1)} | x_{i1}^{(2)} | \ldots | x_{i1}^{(\slotnb)} \rsem  \\
\ctxtrep{cx_{i2}} &= \lsem x_{i2}^{(1)} | x_{i2}^{(2)} | \ldots | x_{i2}^{(\slotnb)} \rsem \\
& \ldots  \\
\ctxtrep{cx_{i\inputlen}} &= \lsem x_{i\inputlen}^{(1)} | x_{i\inputlen}^{(2)} | \ldots | x_{i\inputlen}^{(\slotnb)} \rsem 
\end{aligned}
\qquad \text{Manual Packing of $x_i$}
\label{eqn:CP_attribute_packing}
\end{equation}

To shorten the notation, let $y_j$ denote the threshold of the $j$-th decision node (i.e., $y_j = v_{j}.\thr$) and assume $v_{j}.\aindex = i$. The server just encrypts each threshold bitwise which automatically replicates the bit to all slots. This is illustrated in Equation \ref{eqn:CP_threshold_packing}.
\begin{equation}
\begin{aligned}
\ctxtrep{cy_{j1}} &= \lsem y_{j1} | y_{j1} | \ldots | y_{j1} \rsem \\
\ctxtrep{cy_{j2}} &= \lsem y_{j2} | y_{j2} | \ldots | y_{j2} \rsem \\
& \ldots \\
\ctxtrep{cy_{j\inputlen}} &= \lsem y_{j\inputlen} | y_{j\inputlen} | \ldots | y_{j\inputlen} \rsem
\end{aligned}
\qquad \text{Automatic Packing of $y_j$}
\label{eqn:CP_threshold_packing}
\end{equation}
Note that $(\ctxtrep{cy_{j\inputlen}}, \ldots, \ctxtrep{cy_{j1}}) = \ctxtrep{\bitrep{y_j}}$ holds because of SVCP.
The above described encoding allows to compare $\slotnb$ attribute values together with one threshold. This is possible because the routine $\shecmp$ is compatible with SVCP such that we have: 
\begin{equation}
\begin{aligned}
\shecmp((\ctxtrep{cx_{i\inputlen}}, \ldots, \ctxtrep{cx_{i1}}), (\ctxtrep{cy_{j\inputlen}}, \ldots, \ctxtrep{cy_{j1}})) = \\
(\ctxtrep{b_{ij}^{(1)} | b_{ij}^{(2)} | \ldots | b_{ij}^{(\slotnb)}}, \ctxtrep{b_{ji}^{(1)} | b_{ji}^{(2)} | \ldots | b_{ji}^{(\slotnb)}}),
\end{aligned}
\label{eqn:CP_packing_she_compare}
\end{equation}
where $b_{ij}^{(\slotct)} = [x_{i}^{(\slotct)} > y_{j}]$ and $b_{ji}^{(\slotct)} = [y_{j} > x_{i}^{(\slotct)}]$. This results in a single ciphertext such that the $\slotct$-th slot contains the comparison result between $x_{i}^{(\slotct)}$ and $y_{j}$.

Aggregating decision bits remains unchanged as described in Algorithm \ref{Eval_Paths_Algo}. This results in a packed ciphertext $\ctxtrep{b_v} = \ctxtrep{b_{v}^{(1)} | \ldots | b_{v}^{(\slotnb)}}$ for each leaf $v \in \leafset$, where $b_{v}^{(\slotct)} = 1$ if $x^{(\slotct)}$ classifies to leaf $v$ and $b_{u}^{(\slotct)} = 0$ for all other leaf $u \in \leafset - \{v\}$. 
%

For the classification label $c_{v}$ of a leaf $v \in \leafset$, let $\ctxtrep{\bitrep{c_{v}}} = (\ctxtrep{c_{v|k|}}, \ldots, \ctxtrep{c_{v1}})$ denote the encryption of the bit representation $\bitrep{c_{v}} = c_{v|k|}\ldots c_{v1}$.
To select the correct classification label Algorithm \ref{Finalizing_Algo} is updated as follows. We compute $\ctxtrep{c_{v|k|}} \odot \ctxtrep{b_v}, \ldots, \ctxtrep{c_{v1}} \odot \ctxtrep{b_v}$ for each leaf $v \in \leafset$ and sum them component-wise over all leaves. This results in the encrypted bit representation of the correct classification labels.

\paragraph{Packing Threshold Values}
In this case, the client encrypts a single attribute in one ciphertext, while the server encrypts multiple threshold values in a single ciphertext. Hence, for an attribute value $x_i$, the client generates the ciphertexts as in Equation \ref{eqn:SP_attribute_packing}. Let $m_i$ be the number of decision nodes that compare to the attribute $x_i$ (i.e., $m_i = |\{v_j \in \nodeset: v_j.\aindex = i\}|$). The server packs all corresponding threshold values in $\ceil{\frac{m_i}{s}}$ ciphertext(s) as illustrated in Equation \ref{eqn:SP_threshold_packing}.

\begin{equation}
\begin{aligned}
\ctxtrep{cx_{i1}} &= \lsem x_{i1} | x_{i1}| \ldots | x_{i1} \rsem  \\
\ctxtrep{cx_{i2}} &= \lsem x_{i2} | x_{i2} | \ldots | x_{i2} \rsem \\
& \ldots  \\
\ctxtrep{cx_{i\inputlen}} &= \lsem x_{i\inputlen} | x_{i\inputlen}| \ldots | x_{i\inputlen} \rsem 
\end{aligned}
\qquad \text{Automatic Packing of $x_i$}
\label{eqn:SP_attribute_packing}
\end{equation}

\begin{equation}
\begin{aligned}
\ctxtrep{cy_{j1}} &= \lsem y_{j_11} | \ldots | y_{j_{m_i}1} | \ldots \rsem \\
\ctxtrep{cy_{j2}} &= \lsem y_{j_12} | \ldots | y_{j_{m_i}2} | \ldots \rsem \\
& \ldots \\
\ctxtrep{cy_{j\inputlen}} &= \lsem y_{j_1\inputlen} | \ldots | y_{j_{m_i}\inputlen} | \ldots \rsem
\end{aligned}
\qquad \text{Manual Packing of $y_j$}
\label{eqn:SP_threshold_packing}
\end{equation}
The packing of threshold values allows to compare one attribute value against multiple threshold values together. Unfortunately, we do not have access to the slots while performing homomorphic operation. Hence, to aggregate the decision bits, we make $m_i$ copies of the resulting packed decision bits and shift left each decision bit to the first slot. Then the aggregation of the decision bits and the finalizing algorithm work as in the previous case with the only difference that only the result in the first slot matters and the remaining can be set to 0.


\subsection{Efficient Path Evaluation}

As explained above, the encryption algorithm $\enc$ adds noise to the ciphertext which increases during homomorphic evaluation. While addition of
ciphertexts increases the noise slightly, the multiplication increases it explosively \cite{BrakerskiGV11}. The noise must be kept low enough to prevent incorrect decryption. To keep the noise low, one can either keep the circuit’s depth low enough or use the refresh algorithm. In this section, we will focusing on keeping the circuit depth low.

\begin{definition}[Multiplicative Depth]
Let $f$ be a function and $C_f$ be a boolean circuit that computes $f$ and consists of $AND$-gates or multiplication (modulo 2) gates  and $XOR$-gates or addition (modulo 2) gates . The \emph{circuit depth} of $C_f$ is the maximal length of a path from an input gate to an output gate. The \emph{multiplicative depth} of $C_f$ is the path from an input gate to an output gate with the largest number of of multiplication gates.
\end{definition}

For example, consider the function $f([a_1, \ldots, a_n]) = \Pi_{i=1}^{n}a_i$. A circuit that successively multiplies the $a_i$ has multiplicative depth $n$. However, a circuit that divides the array in two halves, multiplies the elements in each half and finally multiplies the result, has multiplicative depth $\ceil{\frac{n}{2}}+1$. This gives the intuition for the following lemma.

\begin{lemma}[Logarithmic Multiplicative Depth Circuit] 
\label{thr:log_mult_depth}
Let $[a_1, \ldots, a_n]$ be an array of $n$ integers and $f$ be the function defined as follows:
$f([a_1, \ldots, a_n]) = [a'_1, \ldots a'_{\ceil{\frac{n}{2}}}]$
where
$$
a'_{i} =
\begin{cases}
	a_{2i-1} \cdot a_{2i} & \mbox{if} ~ (n \bmod 2 = 0) \vee (i < \ceil{\frac{n}{2}}), \\
	a_n & \mbox{if} ~  (n \bmod 2 = 1) \wedge (i = \ceil{\frac{n}{2}}). 
\end{cases}
$$
Moreover, let $f$ be the iterated function where $f^{i}$ is the $i$-th iterate of $f$ defined as follows:
$$
f^{i}([a_1, \ldots, a_n]) =
\begin{cases}
	[a_1, \ldots, a_n] & \mbox{if} ~ i = 0, \\
	f(f^{i-1}([a_1, \ldots, a_n])) & \mbox{if} ~  i \geq 1.
\end{cases}
$$
The $|n|$-th iterate $f^{|n|}$ of $f$ computes $\Pi_{i=1}^{n}a_i$ and has multiplicative depth $|n|-1$ if $n$ is a power of two and $|n|$ otherwise, where $|n|= \log n$ is the bitlength of $n$:
$$f^{|n|}([a_1, \ldots, a_n])=\Pi_{i=1}^{n}a_i$$
\end{lemma}
\begin{proof}
%
For the proof we consider two cases: $n$ is a power of two (i.e., $n=2^l$ for some $l$), and $n$ is not a power of two.
\paragraph{The Power of Two Case} The proof is inductive. Assume $n = 2^l$, we show by induction on $l$. The base case trivially holds. For the inductive step, we assume the statement holds for $n = 2^l$ and show it holds for $n' = 2^{l+1}$. By dividing the array $[a_1, \ldots, a_{n'}]$ in exactly two halves, the inductive assumption holds for each half. Multiplying the results of both halves concludes the proof.
\paragraph{The Other Case} The proof is constructive. Assume $n$ is not a power of two and let $n''$ be the largest power of two such that $n'' < n$, hence $|n''| = |n|$. We divide $[a_1, \ldots, a_n]$ in two halves $A_1 = [a_1, \ldots, a_{n''}]$ and $A' = [a_{n''+1}, \ldots, a_{n}]$. We do this recursively for $A'$ and get a set of subsets of $[a_1, \ldots, a_n]$ which all have a power of two number of elements. The claim then holds for each subset (from the power of two case above) and $A_1$ has the largest multiplicative depth which is $|n''|-1$. By joining the result from $A_1$ and $A'$, we get the product $\Pi_{i=1}^{n}a_i$ with one more multiplication resulting in a multiplicative depth of $|n''| = |n|$.
\end{proof}

		
		


\begin{figure}[tbp]
	\renewcommand{\figurename}{Algorithm}
	\begin{mdframed}
	\begin{algorithmic}[1]
		\Require {leaves set $\leafset$, decision nodes set $\nodeset$}
		\Ensure {Updated $v.\cmp$ for each $v \in \leafset$}
		\Function{\evalpathe}{$\leafset$, $\nodeset$}
		\ForEach {$v\in \leafset$}
		
		\State \textbf{let} $d$ = number of nodes on the path $(\rootnode \rightarrow v)$
		\State \textbf{let} $\path$ be an empty array of length $d$
		\State $l \gets d$
		\State $w \gets v$
		\While {$w \neq \rootnode$ } \Comment construct path to root
			\State $\path[l] \gets \ctxtrep{w.\cmp}$
			\State $l \gets l - 1$
			\State $w \gets w.\pnode$	
		\EndWhile
		
		\State $\ctxtrep{v.\cmp} \gets$ \Call{\evalmul}{$1, d, \path$}
		\EndFor
		\EndFunction
 	\end{algorithmic}

 	\begin{algorithmic}[1]
		\Require {integers $\from$ and $\tosf$, array of nodes $\path$}
		\Ensure {Product of elements in $\path$}
		\Function{\evalmul}{$\from, \tosf, \path$}
			\If {$\from \geq \tosf$}
				\State \Return $\path[\from]$
			\EndIf
			\State $n \gets \tosf - \from + 1$
			\State $\midsf \gets 2^{|n-1| - 1} + \from - 1$ \Comment $|n|$ bitlength of $n$
			\State $\ctxtrep{\leftsf} \gets $ \Call{\evalmul}{$\from, \midsf, \path$}
			\State $\ctxtrep{\rightsf} \gets $ \Call{\evalmul}{$\midsf+1, \tosf, \path$}
			\State \Return $\ctxtrep{\leftsf} \odot \ctxtrep{\rightsf}$
		\EndFunction
	\end{algorithmic}
	\end{mdframed}
	\caption[Algorithm for Path Aggregation  with log Multiplicative Depth]{Paths Evaluation with log Multiplicative Depth}
	\label{fig:eval_path_log_mult_depth}
\end{figure}

Now, we know that sequentially multiplying comparison results on the path to a leaf results in a multiplicative depth which is linear in the depth of tree and increase the noise explosively. Instead of doing the multiplication sequentially, we will therefore do it in such a way as to preserve a logarithmic multiplicative depth. This is described in Algorithm \ref{fig:eval_path_log_mult_depth}. 
Algorithm \ref{fig:eval_path_log_mult_depth} consists of a main function and a sub-function. The main function $\Call{\evalpathe}{}$ collects for each leaf $v$ encrypted comparison results on the path from the root to $v$ and passes it as an array to the sub-function $\Call{\evalmul}{}$ which is a \emph{divide and conquer} type. The sub-function follows the construction described in the proof of Lemma \ref{thr:log_mult_depth}. It divides the array in two parts (left and right) such that the left part has a power of two number of elements. Then it calls the recursion on the two part and returns the product of their results. 

The two functions in Algorithm \ref{fig:eval_path_log_mult_depth} correctly compute the multiplication of decision bits for each path. While highly parallelizable, it is still not optimal, as each path is considered individually. Since multiple paths in a binary tree share a common prefix (from the root), one would ideally want to handle common prefixes one time and not many times for each leaf. This can be solved using \emph{memoization} technique which is an optimization that stores results of expensive function calls such that they can be used latter if needed. Unfortunately, naive memoization would require a complex synchronization in a multi-threaded environment and  linear multiplicative depth. In the next paragraph, we propose a pre-computation on the tree, that would allow us to have the best of both worlds - multiplication with logarithmic depth along the paths, while reusing the result of common prefixes, thus, avoiding unnecessary work.

\subsection{Improving Path Evaluation with Pre-Computation}
The idea behind this optimization is to use directed acyclic graph which we want to define first.
\begin{definition}[DAG]
	A \emph{directed acyclic graph} (DAG) is a graph with directed edges in which there are no cycles. A vertex $v$ of a DAG is said to be \emph{reachable} from another vertex $u$ if there exists a non-trivial path that starts at $u$ and ends at $v$. The reachability relationship is a partial order $\leq$ and we say that two vertices $u$ and $v$ are ordered as $u \leq v$ if there exists a directed path from $u$ to $v$. 
\end{definition}
We require our DAGs to have a unique maximum element. The edges in the DAG define dependency relation between vertices. 
\begin{definition}[Dependency Graph]
	Let $\dagcfun{}$ be the function that takes two DAGs $G_1, G_2$ and returns a new DAG $G_3$ that connects the maxima of $G_1$ and $G_2$.
	We define the function $\dagfun([a_1, \ldots, a_n])$ that takes an array of integers and returns:
	\begin{itemize}
		\item a graph with a single vertex labeled with $a_1$ if $n = 1$
		\item $\dagcfun(\dagfun([a_1, \ldots, a_{n'}]), \dagfun([a_{n'+1}, \ldots, a_n]))$ if $n > 1$ holds, where $n' = 2^{|n|-1}$ and $|n|$ denotes the bitlength of $n$.
	\end{itemize}
	We call the DAG $G$ generated by $G = \dagfun([a_1, \ldots, a_n])$ a \emph{dependency graph}. For each edge $(a_i,a_j)$ in $G$ such that $i < j$, we say that $a_j$ \emph{depends} on $a_i$ and denote this by adding $a_i$ in the \emph{dependency list} of $a_j$. We require that if $L(j) = [a_{i_1}, \ldots, a_{i_{|L(j)|}}]$ is the dependency list of $a_j$ then it holds $i_1 > i_2 > \ldots i_{|L(j)|}$.
\end{definition}
An example of dependency graph generated by the function $\dagfun([a_1, \ldots, a_n])$ is illustrated in Figure \ref{fig:dagexample} for $n = 4$ and $n = 5$.

\begin{figure}[tbp]
\centering
\begin{tikzpicture}[->,>=stealth',shorten >=1pt,auto,node distance=1.5cm, semithick]
  \tikzstyle{every state}=[draw=black,text=black]

  \node[state]         (a1)              {$a_1$};
  \node[state]         (a2) [right of=a1] {$a_2$};
  \node[state]         (a3) [right of=a2] {$a_3$};
  \node[state]         (a4) [right of=a3] {$a_4$};

  \node at (a1) [above=0.5cm, draw=none] {$[]$};
  \node at (a2) [above=0.5cm, draw=none] {$[a_1]$};
  \node at (a3) [above=0.5cm, draw=none] {$[]$};
  \node at (a4) [above=0.5cm, draw=none] {$[a_3,a_1]$};

  \path (a1) edge [bend right] node {} (a2)
		(a2) edge [bend right = 45] node {} (a4)
        (a3) edge [bend right] node {} (a4);
 \end{tikzpicture}
 \begin{tikzpicture}[->,>=stealth',shorten >=1pt,auto,node distance=1.5cm, semithick]
   \tikzstyle{every state}=[draw=black,text=black]
  \node[state]         (a1)               {$a_1$};
  \node[state]         (a2) [right of=a1] {$a_2$};
  \node[state]         (a3) [right of=a2] {$a_3$};
  \node[state]         (a4) [right of=a3] {$a_4$};
  \node[state]         (a5) [right of=a4] {$a_5$};

  \node at (a1) [above=0.5cm, draw=none] {$[]$};
  \node at (a2) [above=0.5cm, draw=none] {$[a_1]$};
  \node at (a3) [above=0.5cm, draw=none] {$[]$};
  \node at (a4) [above=0.5cm, draw=none] {$[a_3,a_1]$};
  \node at (a5) [above=0.5cm, draw=none] {$[a_4]$};

  \path (a1) edge [bend right] node {} (a2)
		(a2) edge [bend right = 45] node {} (a4)
        (a3) edge [bend right] node {} (a4)
		(a4) edge [bend right] node {} (a5);
\end{tikzpicture}
\caption[Example Dependency Graph]{Dependency Graph for $n=4$ and $n=5$}
\label{fig:dagexample}
\end{figure}
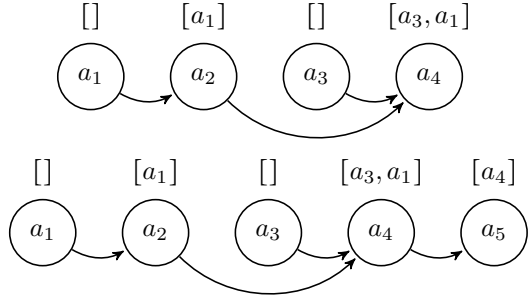

\begin{lemma}
\label{thr:dag_max_def}
Let $[a_1, \ldots, a_n]$ be an array of $n$ integers. Then $\dagfun([a_1, \ldots, a_n])$ as defined above generates a DAG whose maximum element is marked with $a_n$.
\end{lemma}

\begin{lemma}
\label{thr:dag_mult_dep_graph}
Let $[a_1, \ldots, a_n]$ be an array of $n$ integers, $G = \dagfun([a_1, \ldots, a_n])$ be a DAG as defined above and $L(j) = [a_{i_1}, \ldots, a_{i_{|L(j)|}}]$ be the dependency list of $a_j$. 
	\begin{figure}[tbp]
	\renewcommand{\figurename}{Algorithm}
	\centering
	\begin{mdframed}
	\begin{algorithmic}[1]
		\For {$j = 1$ \textbf{to} $j = n$}
			\For {$l = 1$ \textbf{to} $l = |L(j)|$}
				\State $a_j \gets a_j \cdot a_{i_l}$
			\EndFor
		\EndFor
	\end{algorithmic}
	\end{mdframed}
	\caption[Algorithm for Multiplication using Dependency Lists]{Multiplication using Dependency Lists}
	\label{fig:Mult_With_Dep_List}
	\end{figure}
Then Algorithm \ref{fig:Mult_With_Dep_List} computes $\Pi_{i=1}^{n}a_i$ and has a multiplicative depth of $\log(n)$. 
\end{lemma}

The proofs of Lemmas \ref{thr:dag_max_def} and \ref{thr:dag_mult_dep_graph} follow by induction similar to Lemma \ref{thr:log_mult_depth}.  Before describing the improved path evaluation algorithm, we first extend our $\node$ data structure by adding to it a new field representing a stack denoted $\mdag$, that stores the dependency list. Moreover, we group the nodes of the decision tree by level and use an array denoted $\level[]$, such that $\level[0]$ stores a pointer to the root and $\level[i]$ stores pointers to the child nodes of $\level[i-1]$ for $i\geq 1$. Now, we are ready to describe the improved path evaluation algorithm which consists of a pre-computation step and an online step.

The pre-computation is a one-time computation that depends only on the structure of the decision tree and requires no encryption. As described in Algorithm \ref{alg:precom_multDAG}, its main function $\Call{\computedag}{}$ uses the leveled structure of the tree and the dependency graph defined above to compute the dependency list of each node in the tree (i.e., the DAG defined above). The sub-function $\Call{\addedge}{}$ is used to actually add nodes to the dependency list of another node (i.e., by adding edges between these nodes in the DAG). 

The online step is described in Algorithm \ref{alg:evalpathmultDAG}. It follows the idea of Algorithm \ref{fig:Mult_With_Dep_List} by multiplying decision bit level-wise depending on the dependency lists. The correctness follows from Lemma \ref{thr:dag_mult_dep_graph}.

\begin{figure}[tbp]
	\renewcommand{\figurename}{Algorithm}
	\begin{mdframed}
	\begin{algorithmic}[1]
		\Require {integers $\up$ and $\low$}
		\Ensure {Computed $v.\mdag$ for each $v \in \nodeset \cup \leafset$}
		\Function{\computedag}{$\up, \low$}
			
			\If {$\up \geq \low$}  
				\State \Return \Comment end the recursion
			\EndIf
			\State $\eta \gets \low - \up + 1$
			\State $\midsf \gets 2^{|\eta-1|-1} - 1 + \up$ \Comment $|\eta|$ bitlength of $\eta$ 
			\ForEach {$v \in \level[\low]$} 
				\State $\Call{\addedge}{v, \low, \midsf}$
			\EndFor
			\For {$i = \midsf + 1$ \textbf{to} $\low - 1$} \Comment non-deepest leaves
				\ForEach {$v \in \level[i] \cap \leafset$}  
					\State $\Call{\addedge}{v, i, \midsf}$
				\EndFor
			\EndFor
			\State $\Call{\computedag}{\up, \midsf}$
			\State $\Call{\computedag}{\midsf + 1, \low}$
		\EndFunction
	\end{algorithmic}

	\begin{algorithmic}[1]
		\Require {Node $v$, integers $\currLvl$ and $\destLvl$}
		\Ensure {Updated $v.\mdag$}
		\Function{\addedge}{$v$, $\currLvl$, $\destLvl$}
			\State $w \gets v$
			\While {$\currLvl > \destLvl$}
				\State $w \gets w.\pnode $
				\State $\currLvl \gets \currLvl - 1$
			\EndWhile
			\State $v.\mdag.\push(w)$ \Comment $\mdag$ is a stack 
		\EndFunction
	\end{algorithmic}
\end{mdframed}
	\caption[Algorithm for Pre-computation of Multiplication DAG]{Pre-computation of Multiplication DAG}
 	\label{alg:precom_multDAG}
\end{figure}

\begin{figure}[tbp]
	\renewcommand{\figurename}{Algorithm}
	\begin{mdframed}
	\begin{algorithmic}[1]
		\Require {set of nodes stored by level in array $\level$ }
		\Ensure {Updated $v.\cmp$ for each $v \in \leafset$}
		\Function{\evalpathp}{} 
		
		\For {$i = 1$ \textbf{to} $d$} \Comment from top to bottom level
			\ForEach {$v \in \level[i]$}
			\While {$v.\mdag.\emptyQ() = false$} \Comment $\mdag$ = stack
				\State $w \gets v.\mdag.\pop()$
				\State $\ctxtrep{v.\cmp} \gets \ctxtrep{v.\cmp} \odot \ctxtrep{w.\cmp}$
			\EndWhile
			\EndFor	
		\EndFor
		
		\EndFunction
	\end{algorithmic}
\end{mdframed}
	\caption[Algorithm for Aggregating Decision Bits With Pre-Computation]{Aggregate Decision Bits with precomputed DAG}
	\label{alg:evalpathmultDAG}
\end{figure}

\section{Arithmetic Implementation}
\label{Arithmetic_Implementation}

In this section, we describe $\pdteint$, an instantiation of the basic scheme that encodes the plaintexts such that the computation is done using an arithmetic circuit. This means that a ciphertext now encrypts an integer and that arithmetic operations are no longer $\bmod~2$, but $\bmod~2^{l}$ for some $l > 2$.

\subsection{Arithmetic Integer Comparison}

We first describe our modified version of the Lin-Tzeng comparison protocol~\cite{LinT05}.
The main idea of their construction is to reduce the greater-than comparison to the set intersection problem of prefixes. Let $x_i$ and $y_j$ be inputs of client and server, respectively, with the goal to compute $[x_i > y_j]$.
\paragraph{Input Encoding}
Let $\textsc{Int}(z_{\eta}\cdots z_{1}) = z$ be a function that takes a bit string of length $\eta$ and parses it into the $\eta-$bit integer $z = \sum_{l=1}^{\eta}z_{l} \cdot 2^{l-1}$. 
The \emph{0-encoding} $V_{x_i}^{0}$ and \emph{1-encoding} $V_{x_i}^{1}$ of an integer input $x_i$ are the following vectors:
$V_{x_i}^{0} = (v_{i\inputlen}, \cdots, v_{i1}), V_{x_i}^{1} = (u_{i\inputlen}, \cdots, u_{i1})$, such that $ \forall l \in \{1 \ldots \inputlen\}$  
\[
v_{il}=
		\begin{cases}
		\textsc{Int}(x_{i\inputlen}x_{i\inputlen-1}\cdots x_{il'}1) & \mbox{if} ~ x_{il} = 0\\  
		r^{(0)}_{il}                         & \mbox{if} ~ x_{il} = 1
		\end{cases} 
\]

\[		
u_{il}=
		\begin{cases}
		\textsc{Int}(x_{i\inputlen}x_{i\inputlen-1}\cdots x_{il}) & \mbox{if} ~ x_{il} = 1\\  
		r^{(1)}_{il}                      & \mbox{if} ~ x_{il} = 0,
		\end{cases} 
\]

\noindent where $l'=l+1$, and $r^{(0)}_{il}$, $r^{(1)}_{il}$ are random numbers of a fixed bitlength $\nu > \inputlen$ (e.g. $2^{\inputlen} \leq r^{(0)}_{il}, r^{(1)}_{il} < 2^{\inputlen+1}$) with $LSB(r^{(0)}_{il}) = 0$ and  $LSB(r^{(1)}_{il}) = 1$ (LSB is the least significant bit). If the \textsc{Int} function is used the compute the element at position $l$, then we call it a \emph{proper encoded element} otherwise we call it a \emph{random encoded element}. Note that a random encoded element  $r^{(1)}_{il}$ at position $l$ in the 1-encoding of $x_i$ is chosen such that it is guaranteed to be different to a proper or random encoded element at position $l$ in the 0-encoding of $y_j$, and vice versa. Hence, it enough if $r^{(1)}_{il}$ and $r^{(0)}_{il}$ are one or two bits longer than any possible proper encoding element at position $l$. Also note that the bitstring $x_{i\inputlen}x_{i\inputlen-1}\cdots x_{il}$ is interpreted by the function $\textsc{Int}$ as a bitstring $z_{\inputlen-l+1}\cdots z_{1}$ with length $\mu-l+1$ where $z_{1} = x_{il}, z_{2} = x_{i(l+1)}, \ldots, z_{\mu-l+1} = x_{i\mu}$.
If we see $V_{x_i}^{0}, V_{y_j}^{1}$ as sets, then  $x_i > y_j$ iff they have exactly one common element.

\begin{lemma}
\label{LinTzengLemma}
 Let $x_i$ and $y_j$ be two integers, then $x_i > y_j ~\mbox{iff}~  V = V_{x_i}^{1} - V_{y_j}^{0}$ has a unique position with 0.
\end{lemma}




\paragraph{The Protocol}
Let $\lsem V_{x_i}^{0} \rsem = \lsem v_{i\inputlen} \rsem, \ldots, \lsem v_{i1} \rsem$ (respectively $\lsem V_{x_i}^{1} \rsem = \lsem u_{i\inputlen} \rsem, \ldots, \lsem u_{i1} \rsem$) denote the componentwise encryption of $V_{x_i}^{0}$ (resp. $V_{x_i}^{1}$). The client sends $\lsem V_{x_i}^{0} \rsem, \lsem V_{x_i}^{1} \rsem$ to the server. To determine the comparison result for $x_i > y_j$, the server evaluates the function $\textsc{LinCompare}(\lsem V_{x_i}^{1} \rsem, \lsem V_{y_j}^{0} \rsem)$ (Algorithm \ref{Lin_Compare_Algo}) which returns $\inputlen$ ciphertexts among which exactly one encrypts zero if an only if $x_i > y_j$. For the decision tree evaluation, the server omits the randomization in Step \ref{LTT_rand_step} and the random permutation in Step \ref{LTT_perm_choice_step}, since this not the final result. Moreover, the server collects the difference ciphertexts $c_l$ in an array and uses the multiplication algorithm with logarithmic multiplicative depth.

\paragraph{Difference to the original protocol}
In contrast to the original protocol of Lin and Tzeng \cite{LinT05}, we note the following differences:
\begin{itemize}
	\item Additively HE instead of multiplicative: As explained above multiplication increases the noise exponentially while addition increases it only linearly.
	\item The \textsc{Int} function: Instead of relying on a collision-free hash function as Lin and Tzeng \cite{LinT05}, we use the \textsc{Int} function which is simpler to implement and more efficient as it produces smaller values.
	\item The choice of random encoded elements $r^{(0)}_{il}$, $r^{(1)}_{il}$: We choose the random encoded elements as explained above and encrypt them, while the original protocol uses ciphertexts chosen randomly in the ciphertext space. 
	\item Encrypting the encodings on both side: In the original protocol, the evaluator has access to $y_j$ in plaintext and does not need to choose random encoded elements. By encoding as explained in our modified version, we can encrypt both encodings and delegate the evaluation to a third party which is not allowed to have access to the inputs in plaintext.
	\item Aggregation: The multiplication of the ciphertexts returned by Algorithm \ref{Lin_Compare_Algo} returns a ciphertext encrypting either 0 or a random number. 
\end{itemize}


\begin{figure}[tbp]
	\renewcommand{\figurename}{Algorithm}
	\begin{mdframed}
		\begin{algorithmic}[1]
		\Function {\textsc{LinCompare}}{$\lsem V_{x_i}^1 \rsem, \lsem V_{y_j}^0 \rsem$}
				\State \textbf{parse} $\lsem V_{x_i}^1 \rsem$ \textbf{as} $\lsem u_{i\inputlen} \rsem, \ldots, \lsem u_{i1} \rsem$
				\State \textbf{parse} $\lsem V_{y_j}^0 \rsem$ \textbf{as} $\lsem v_{i\inputlen} \rsem, \ldots, \lsem v_{i1} \rsem$
				\For {$l := 1$ \textbf{to}  $\inputlen$} 
					\State \textbf{choose} a random $r_l$ from the plaintext space \label{rand_choice_step}
					\State $c_{l} = \lsem (u_{il} - v_{jl}) \cdot r_l \rsem$ \label{LTT_rand_step}
				\EndFor 
				\State \textbf{choose} a random permutation $\pi$ \label{LTT_perm_choice_step} 
				\State  \Return $\pi(c_{\inputlen}, \cdots, c_{1})$
		\EndFunction
		\end{algorithmic}
	\end{mdframed}
	\caption[Modified Lin-Tzeng Comparison Protocol]{Modified Lin-Tzeng Comparison Protocol}
	\label{Lin_Compare_Algo}
\end{figure}
The modified comparison algorithm as used for PDTE is illustrated in Algorithm \ref{Lin_Compare_Algo_DT}.
Note that, this can be computed using binary gates as well, by encrypting the 0/1-encodings binary-wise resulting in $\inputlen$ blocks of ciphertexts, computing XOR-gates in parallel for each block, then computing  OR-gates in parallel for each block and finally summarizing the results using AND-gates. The multiplicative depth will be $2\inputlen$.

\subsection{Arithmetic PDTE Protocol}

In this section, we use the modified Lin-Tzeng comparison explained above for the decision tree evaluation. We follow the structure of the basic protocol as describe in Protocol \ref{Basic_PDTE_Algo}.

\paragraph{Encrypting the Atttribute Values} 
The protocol starts with the client encrypting and sending its input to the server. For each attribute value $x_i$ the client sends the encryptions $\lsem V_{x_i}^{0} \rsem = \lsem v_{i\inputlen} \rsem, \ldots, \lsem v_{i1} \rsem$ and $\lsem V_{x_i}^{1} \rsem = \lsem u_{i\inputlen} \rsem, \ldots, \lsem u_{i1} \rsem$) of the 0-encoding $V_{x_i}^{0}$ and 1-encoding $V_{x_i}^{1}$ of $x_i$. Note that, this is still compatible with the trusted randomizer technique, where we will use sequences of integers instead of bit strings.

\paragraph{Evaluating Decision Nodes and Paths}
Let $y_j$ be the threshold of a decision node that compares to $x_i$. We assume that $x_i \neq y_j$ for all $i,j$. The parties can ensure this by having the client adding a bit 0 the bit representation of each $x_i$, and the server adding a bit 1 to the bit representation of each $y_j$ before encoding the values. Then from the definition of the tree evaluation, we move to the right if $[x_i \geq y_j]$ or the left otherwise. This is equivalent of testing $[x_i > y_j]$ or $[y_j > x_i]$, since we assume $x_i \neq y_j$. Therefore, for each decision node $y_j$ with corresponding attribute $x_i$, the server uses $\Call{LinCompareDT}{\lsem V_{x_i}^1 \rsem, \lsem V_{y_j}^0 \rsem}$ to mark the edge right to $y_j$ and $\Call{LinCompareDT}{\lsem V_{y_j}^1 \rsem, \lsem V_{x_i}^0 \rsem}$ to mark the edge left to $y_j$. As a result, one edge will be marked with a ciphertext of 0, while the other will be marked with a ciphertext of a random plaintext. It follows that the sum of marks along each path of the tree, will result to an encryption of 0 for the classification path and an encryption of a random plaintext for other paths.

\paragraph{Computing the Result's Ciphertext}
To reveal the final result to the client, we do the following. For each  ciphertext $\ctxtrep{\cost_v}$ of Algorithm \ref{fig:Arithmetic_PDTE}, the server chooses a random number $r_v$, computes $\ctxtrep{result_v} \gets \ctxtrep{\cost_v \cdot r_v + v.\clabel}$ and sends the resulting ciphertexts to the client in a random order.  Alternatively, the server can make a trade-off between communication and computation by using the shift operation to pack many $result_v$ in a single ciphertext.

\paragraph{Using Ciphertext Packing}
Recall that our modified Lin-Tzeng comparison requires only component-wise subtraction and a multiplication of all components. Therefore, the client can pack the 0-encoding of each $x_i$ in one ciphertext and sends $\ctxtrep{v_{i\inputlen} | \ldots | v_{i1} | 0 | \ldots | 0}$ instead of $\ctxtrep{V_{x_i}^{0}}$ (and similar for the 1-encoding). Then the server does the same for each threshold value and evaluates the decision node by computing the differences $\ctxtrep{d_{ij}} \gets \ctxtrep{u_{i\inputlen}-v_{j\inputlen} | \ldots | u_{i1}-v_{j1} | 0 | \ldots | 0}$ with one homomorphic subtraction. To multiply the $\inputlen$ relevant components in $\ctxtrep{d_{ij}}$, we use $|\inputlen|$ (bitlength of $\inputlen$) left shifts and $|\inputlen|$ multiplications to shift $\Pi_{l=1}^{\inputlen}(u_{il}-v_{jl})$ to the first slot. The path evaluation and the computation of the result's ciphertext remain as explained above.
We also note that the packing of attribute values and the packing of threshold values work similar to the binary implementation of Section \ref{Binary_Implementation}.

\begin{figure}[tbp]
	\renewcommand{\figurename}{Algorithm}
	\begin{mdframed}
		\begin{algorithmic}[1]
		\Function {\textsc{LinCompareDT}}{$\lsem V_{x_i}^1 \rsem, \lsem V_{y_j}^0 \rsem$}
				\State \textbf{parse} $\lsem V_{x_i}^1 \rsem$ \textbf{as} $\lsem u_{i\inputlen} \rsem, \ldots, \lsem u_{i1} \rsem$
				\State \textbf{parse} $\lsem V_{y_j}^0 \rsem$ \textbf{as} $\lsem v_{i\inputlen} \rsem, \ldots, \lsem v_{i1} \rsem$
				\State \textbf{let} $arr$ be an empty array of size $\inputlen$ 
				\For {$l := 1$ \textbf{to}  $\inputlen$} 
					\State $arr[l] \gets \lsem u_{il} - v_{jl} \rsem$ \label{LTT_DT_enqueue_step}
				\EndFor
				\State  \Return $\Call{\evalmul}{1, \inputlen, arr}$ \label{LTT_DT_LogMlt_step}
		\EndFunction
		\end{algorithmic}
	\end{mdframed}
	\caption[Modified Lin-Tzeng Protocol for PDTE]{Modified Lin-Tzeng Protocol for PDTE}
	\label{Lin_Compare_Algo_DT}
\end{figure}

\begin{figure}[tbp]
	\renewcommand{\figurename}{Algorithm}
	\centering
	\begin{mdframed}
	\begin{algorithmic}[1]
		\State $\ctxtrep{\rootnode.\cmp} \gets \ctxtrep{0}$
		\ForEach {$v \in \nodeset$}
			\State $\ctxtrep{v.\rnode.\cmp} \gets \Call{LinCompare}{\lsem V_{x_i}^1 \rsem, \lsem V_{y_j}^0 \rsem}$
			\State $\ctxtrep{v.\lnode.\cmp} \gets \Call{LinCompare}{\lsem V_{y_j}^1 \rsem, \lsem V_{x_i}^0 \rsem}$
		\EndFor

		\ForEach {$v \in \leafset$}
			\State \textbf{let} $P_v$ be the array of nodes on the path $(\rootnode \rightarrow v)$
			\State $\ctxtrep{\cost_v} \gets \ctxtrep{0}$
			\ForEach {$u \in P_v$}
				\State $\ctxtrep{\cost_v} \gets \ctxtrep{\cost_v} \oplus \ctxtrep{u.\cmp}$
			\EndFor
		\EndFor
	\end{algorithmic}
	\end{mdframed}
	\caption[Arithmetic PDTE Algorithm]{Arithmetic PDTE Algorithm}
	\label{fig:Arithmetic_PDTE}
\end{figure}

\section{Evaluation}
\label{Evaluation}
In this section, we discuss some implementation details and evaluate our schemes.

\subsection{Implementation Details}
We implemented our algorithms using HElib \cite{HaleviS14} and TFHE \cite{TFHE, ChillottiGGI18}. 
HElib is a C++ library that implements FHE. The current version includes an implementation of the leveled FHE BGV scheme \cite{BrakerskiGV11}. HElib also includes various optimizations that make FHE runs faster, including the Smart-Vercauteren ciphertext packing (SVCP) techniques \cite{SmartV2014}. 

TFHE is a C/C++ library that implements FHE proposed by Chillotti et al.~\cite{ChillottiGGI16, ChillottiGGI17}. It allows to evaluate any boolean circuit on encrypted data. The current version  implements a very fast gate-by-gate bootstrapping, i.e., bootstrapping is performed after each gate evaluation. Future versions will include leveled FHE and ciphertext packing as described by Chillotti et al.~\cite{ChillottiGGI17a}. Dai and Sunar \cite{cuFHE, dai2015cuhe} propose an implementation of TFHE on CUDA-enabled GPUs that is 26 times faster.

We evaluated our implementation on an AWS instance with Intel(R) Xeon(R) Platinum 8124M CPU @ 3.00GHz running Ubuntu 18.04.2 LTS. The Instance has 36 CPUs, 144 GB Memory and 8 GB SSD. As the bottleneck of our scheme is the overhead of the homomorphic computation, we focus on the computation done by the server. We start by generating appropriate encryption parameters and evaluating the performance of basic operations.

\subsection{Basic Operations}
Recall that our plaintext space is usually a ring $\mathbb{Z}_q[X]/(X^{\poldegree}+1)$ and that the encryption scheme might be a leveled FHE with parameter $\nblevel$. For HElib, the parameters $\poldegree$ and $\nblevel$ determines how to generate encryption keys for a security level $\secprm$ which is at least 128 in all our experiments. Table \ref{tab:tablehelibcontext} summarizes the parameters we used for key generation and the resulting sizes for encryption keys and ciphertexts.
We will refer to it as \emph{homomorphic context} or just \emph{context}.
For HElib, one needs to choose $L$ large enough than the depth of the circuit to be evaluated and then computes an appropriate value for $\poldegree$ that ensures a security level at least 128. 
We experimented with tree different contexts ($\helibsmall, \helibmed, \helibbig$) for the binary representation used in $\pdtebin$ and the context $\helibint$ for the integer representation used in $\pdteint$. 
For TFHE, the default value of $\poldegree$ is 1024 and the security level can be chosen up to 128 while $\nblevel$ is infinite because of the gate-by-gate bootstrapping. We used the context $\tfheIZB$ to evaluate $\pdtebin$ with TFHE. Table \ref{tab:encdec_runtime} reports the average runtime for encryption and decryption over 100 runs. The columns \enquote{Enc Vector} and \enquote{Dec Vector} stand for encryption and decryption using SIMD encoding and decoding, which is not supported by TFHE yet.

\begin{table}[tbp]
	\begin{center}
		
		\begin{tabular}{c| c | c | c| c | c | c | c }
			\textbf{Name}&\textbf{$\nblevel$} & \textbf{$\poldegree$} & \textbf{$\secprm$}& \textbf{Slots} &\textbf{\sk} & \textbf{\pk} & \textbf{Ctxt} \\ 
			& &  & (bits)&  & (MB) & (MB) & (MB)  \\ 
			\hline
			$\helibsmall$ &200  & 13981 & 151& 600 & 52.2 & 51.6 & 1.7 \\
			\hline			
			$\helibmed$ &300 & 18631 & 153& 720 & 135.4  & 134.1 & 3.7 \\
			\hline
			$\helibbig$&500 & 32109 & 132& 1800 &  370.1& 367.1 & 8.8 \\
			\hline
			$\helibint$&450 & 24793 & 138.161& 6198 &  370.1& 367.1 & 8.8 \\
			\hline
			$\tfheIZB$ & $\infty$ & 1024 &128& 1 &  82.1& 82.1 & 0.002			
		\end{tabular}
		\caption{Key Generation's Parameters and Results}
		\label{tab:tablehelibcontext}
	\end{center}
\end{table}

\begin{table}[tbp]
	\begin{center}
		
		\begin{tabular}{ c | c  c | c   c }
			HElib & Enc Single & Enc Vector & Dec Single & Dec Vector\\ 
			Context& (ms) & (ms)  & (ms) & (ms) \\ 
			\hline 
			$\helibsmall$ &59.21& 59.41 & 26.08 & 26.38 \\
			$\helibmed$  & 124.39& 124.93 &54.31 & 54.92\\
			$\helibbig$ & 283.49 & 284.31 & 127.11 & 128.32 \\
			$\helibint$ & 323.41 & 488.77 & 88.63 & 93.50 \\
			$\tfheIZB$ & 0.04842 & n/a & 0.00129 & n/a \\
		\end{tabular}
		\caption{Encryption/Decryption Runtime}
		\label{tab:encdec_runtime}
	\end{center}
\end{table}

\subsection{Homomorphic Operations in HElib}
The opposite to the notion of \emph{ciphertext noise} is the notion of \emph{ciphertext capacity} or just \emph{capacity} which is also determined by $\nblevel$ and estimates the \emph{capacity} of a ciphertext to be used in homomorphic operations. In Figure \ref{fig:helib_add_noise} and \ref{fig:helib_mul_noise}, we reported the remaining capacity of a ciphertext after a number of consecutive additions or multiplications starting from the values of $\nblevel$ in Table \ref{tab:tablehelibcontext}. They show that the capacity is reduced only slightly after addition, but explosively after multiplication. Note that the encryption operation already has an impact on $\nblevel$. Figure \ref{fig:log_mul} shows that doing the multiplication with logarithmic depth (Lemma \ref{thr:log_mult_depth}) reduced the capacity sublinearly instead of linearly as in Figure \ref{fig:helib_mul_noise}. The sublinear complexity is also illustrated in Figure \ref{fig:cmp_isolated_bench} for the comparison circuit which also has a logarithmic multiplicative depth \cite{CheonKK15, CheonKK16, CheonKL15}. Figure \ref{fig:runtimeaddmul} illustrates runtimes (best and average) for addition and multiplication in HElib over 100 runs showing that homomorphic addition is really fast compare to multiplication. We report the runtime for the comparison circuit in Figure \ref{fig:cmp_runtime} comparing runtime for HElib and TFHE. While both are linear in the bitlength, the runtime for HElib increases very quickly.

\begin{figure}[tbp]
	\begin{center} 
		\includegraphics[width=.4\textwidth]{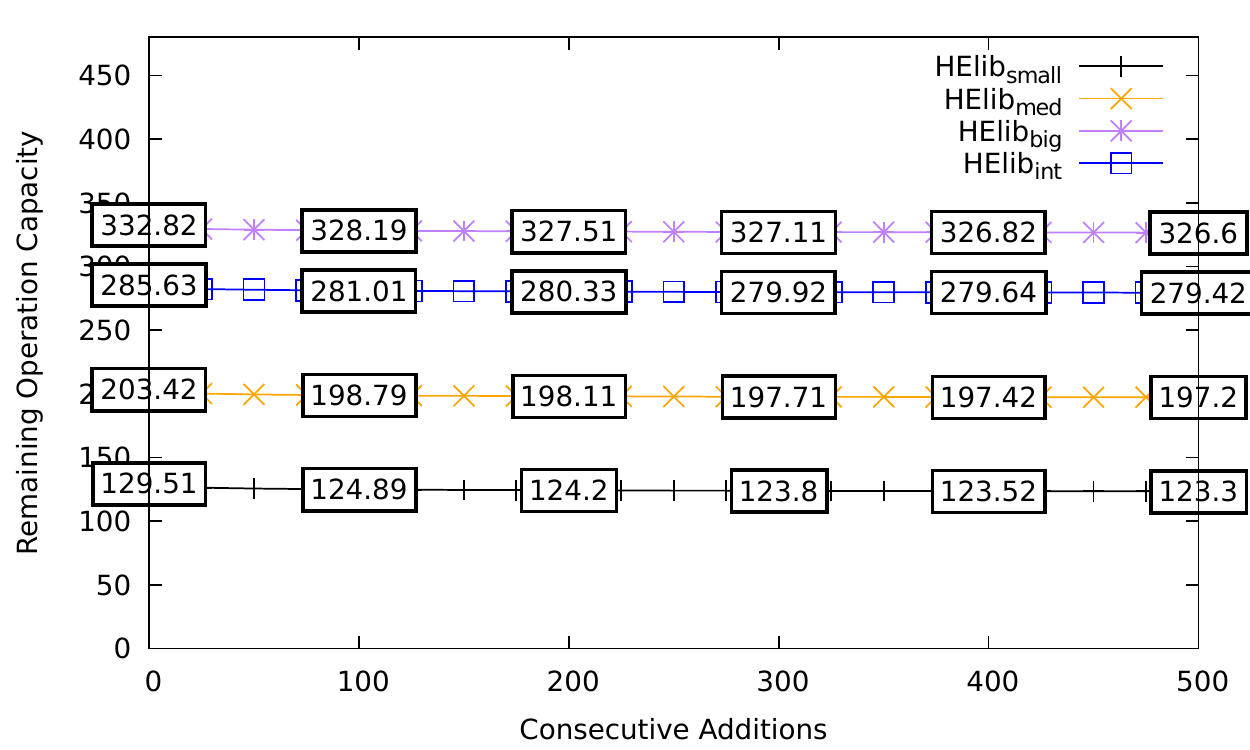}
		\caption{Capacity after consecutive additions}
		\label{fig:helib_add_noise}
	\end{center}
\end{figure}

\begin{figure}[tbp]
	\begin{center} 
		\includegraphics[width=.4\textwidth]{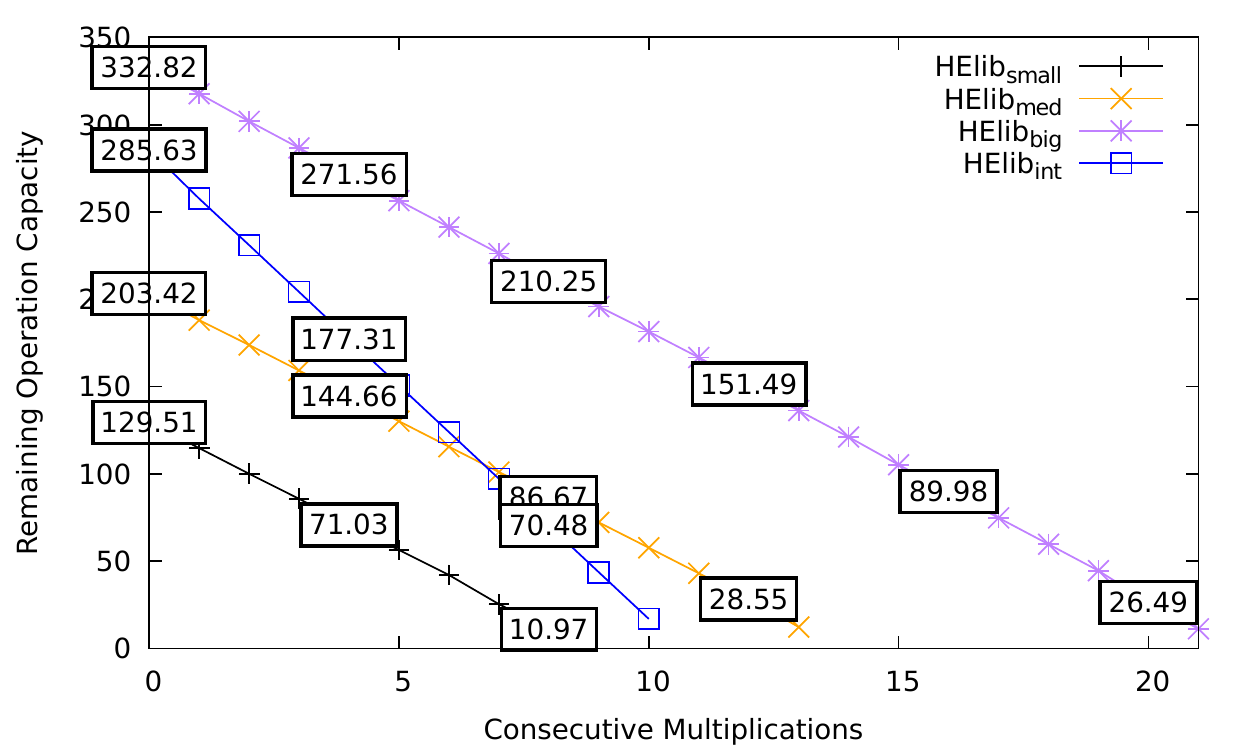}
		\caption{Capacity after consecutive multiplications}
		\label{fig:helib_mul_noise}
	\end{center}
\end{figure}

\begin{figure}[tbp]
	\begin{center} 
		\includegraphics[width=.4\textwidth]{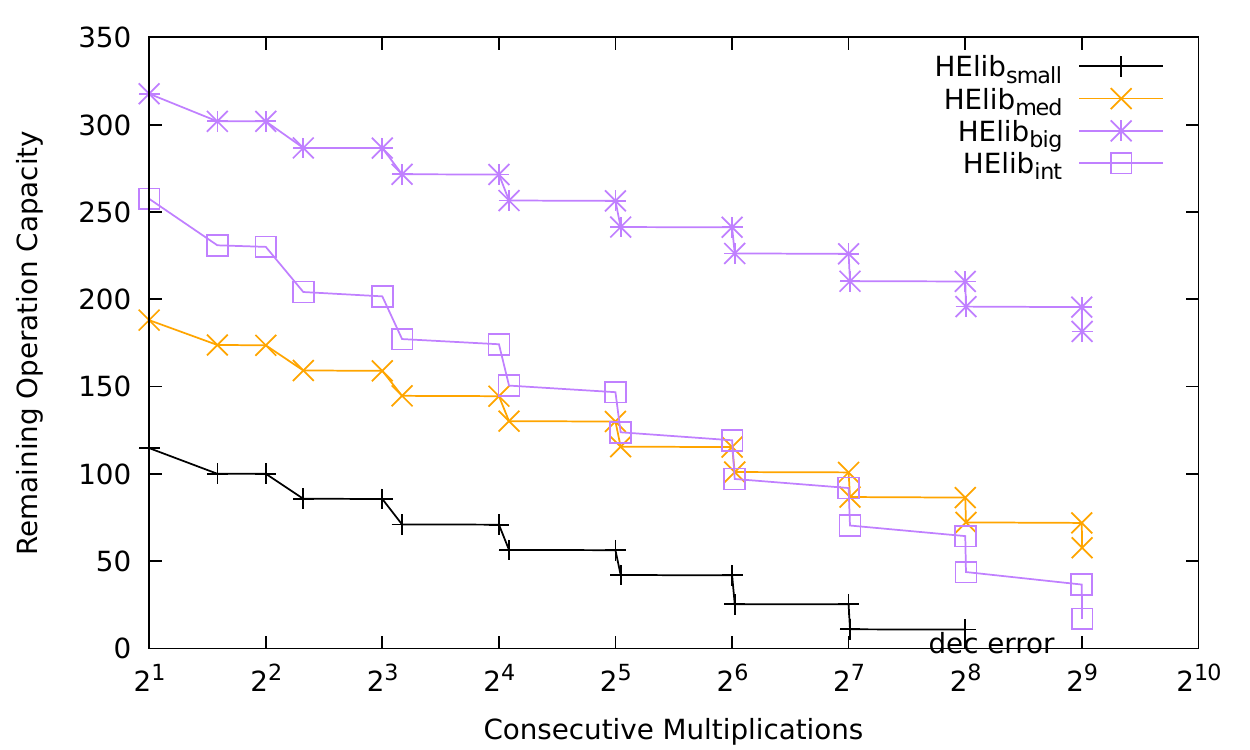}
		\caption{Capacity after multiplication with $\log$ Depth}
		\label{fig:log_mul}
	\end{center}
\end{figure}

\begin{figure}[tbp]
	\begin{center} 
		\includegraphics[width=.4\textwidth]{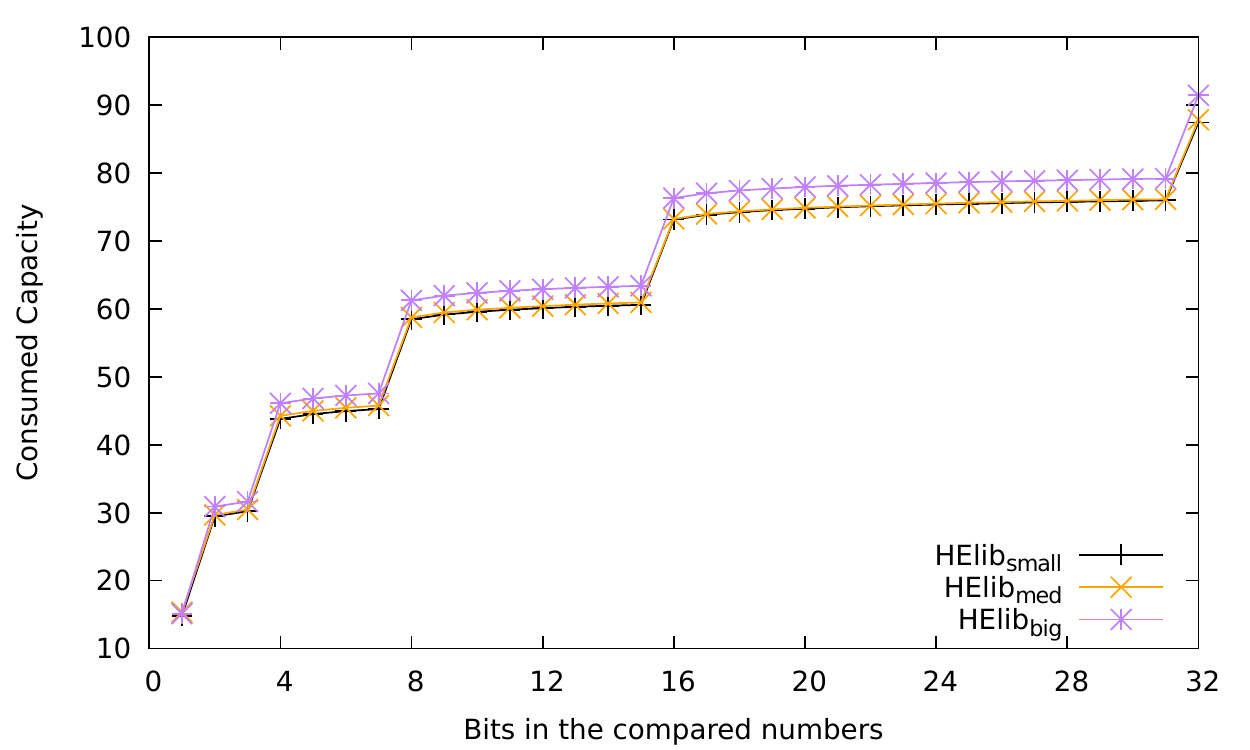}
		\caption{Comparison Capacity Consumption in HElib}
		\label{fig:cmp_isolated_bench}
	\end{center}
\end{figure}

\begin{figure}[tbp]
	\begin{center} 
		\includegraphics[width=.5\textwidth]{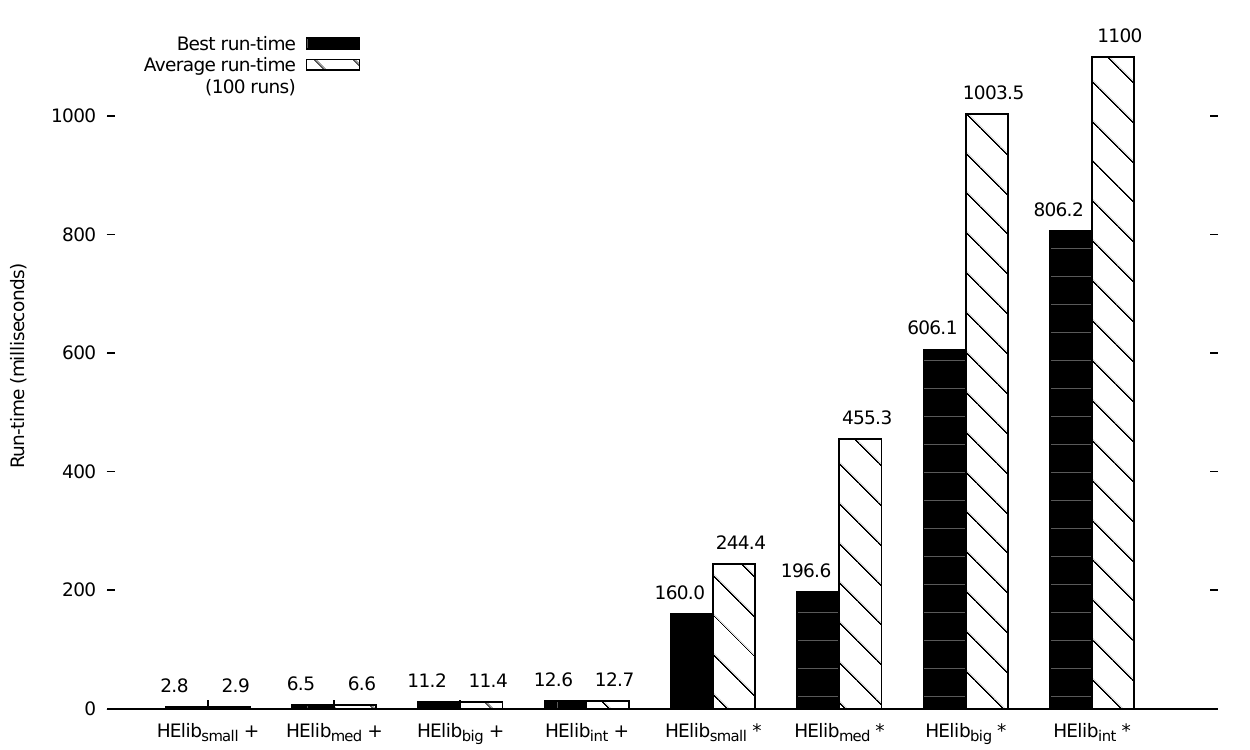}
		\caption{Runtime for Addition and Multiplication in HElib}
		\label{fig:runtimeaddmul}
	\end{center}
\end{figure}

\begin{figure}[tbp]
	\includegraphics[width=.4\textwidth]{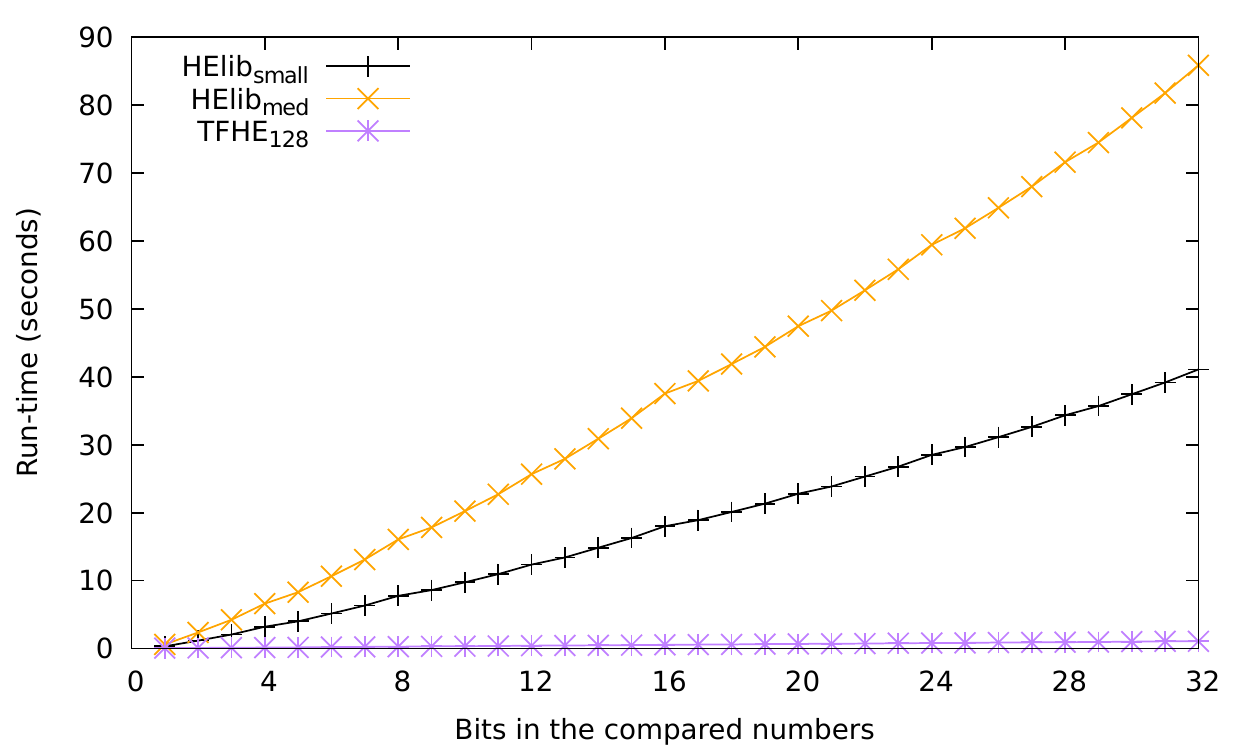}
	\caption{Comparison Run-time Cost}
	\label{fig:cmp_runtime}
\end{figure}

\subsection{Performance of \pdtebin}
In this section, we report on our experiment with $\pdtebin$ on complete trees. Recall that for FHE supporting SIMD, we can use attribute values packing that allows to evaluate many attribute vectors together. We, therefore, focus on attribute packing to show the advantage of SIMD. Figure \ref{fig:classify_one_input} illustrates the amortized runtime of $\pdtebin$ with HElib. That is, the time of one PDTE evaluation divided by the number of slots provided by the used homomorphic context. As one can expect, the runtime clearly depends on the bitlength of the attribute values and the depth of the tree. The results show a clear  advantage of HElib when classifying large data sets.
For paths aggregation, we proposed $\evalpathe$ (Algorithm \ref{fig:eval_path_log_mult_depth}) and $\evalpathp$ (Algorithm \ref{alg:evalpathmultDAG}). Figure \ref{fig:dag_vs_separate} illustrates $\pdtebin$ runtime using these algorithms in a multi-threaded environment and shows a clear advantage of $\evalpathp$ which will be used in the remaining experiments with $\pdtebin$.
Figure \ref{fig:relation_16bit} illustrates the runtime of $\pdtebin$ with $\helibmed$ showing that the computation cost is dominated by the computation of decision bits which involves homomorphic evaluation of comparison circuits. In Figure \ref{fig:pdte_tfhe}, we report the evaluation of $\pdtebin$ using TFHE, which shows a clear advantage compare to HElib. For the same experiment with 72 threads, TFHE evaluates a complete tree of depth 10 and 64-bit input in less than 80 seconds, while HElib takes about 400 seconds for 16-bit input. Recall that, a CUDA implementation \cite{dai2015cuhe, cuFHE} of TFHE can further improve the time of $\pdtebin$ using TFHE.
\begin{figure}[tbp]
 	\begin{center} 
 		\includegraphics[width=.4\textwidth]{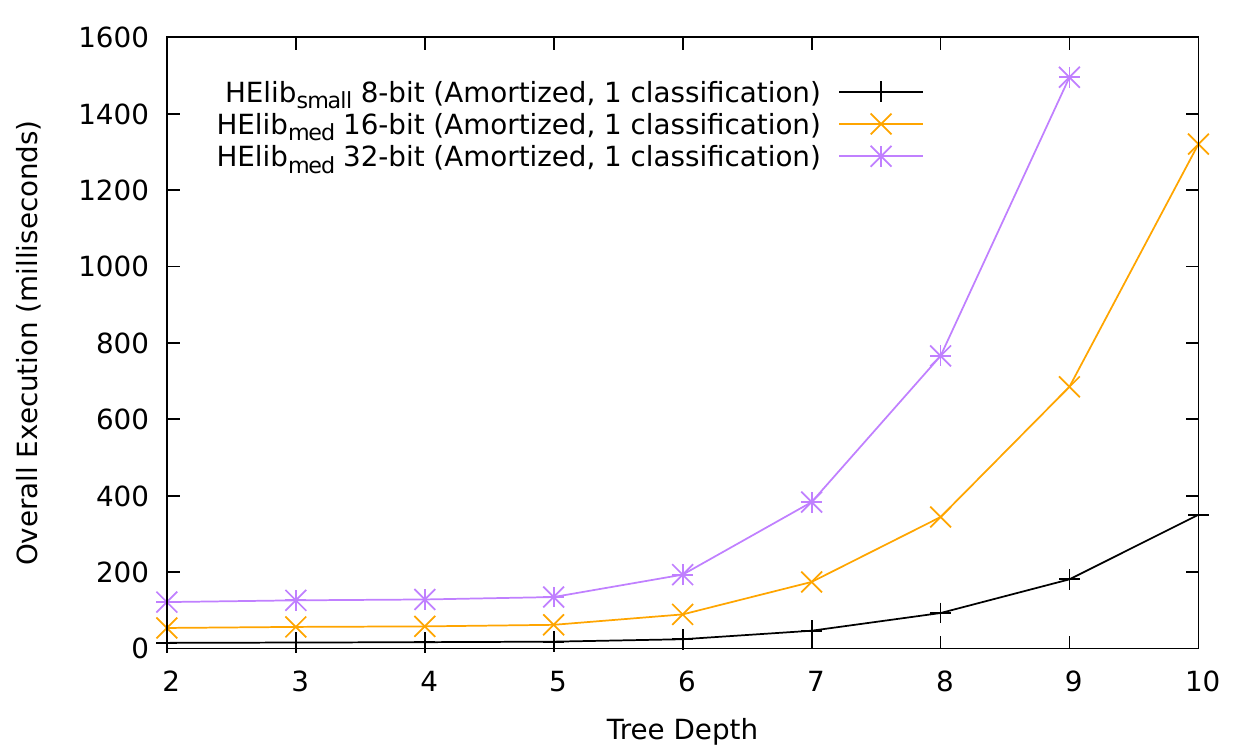}
 		\caption{Amortized $\pdtebin$ Runtime with HElib} 
 		\label{fig:classify_one_input}
 	\end{center}
\end{figure}

 \begin{figure}[tbp]
	\begin{center} 
		\includegraphics[width=.4\textwidth]{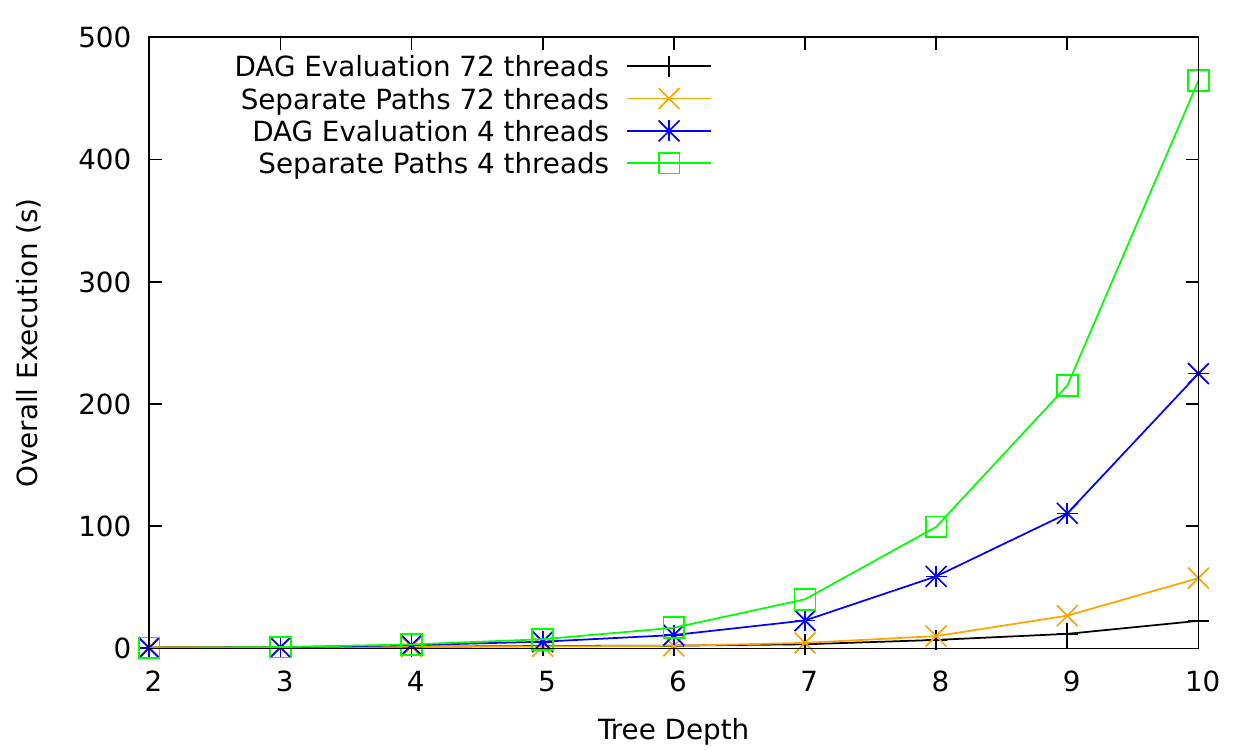}
		\caption{$\pdtebin$ Runtime with $\helibsmall$ Comparing $\evalpathp$ (DAG) vs. $\evalpathe$ (Separate Paths)}
		\label{fig:dag_vs_separate}
	\end{center}
\end{figure}

\begin{figure}[tbp]
	\begin{center} 
		\includegraphics[width=.4\textwidth]{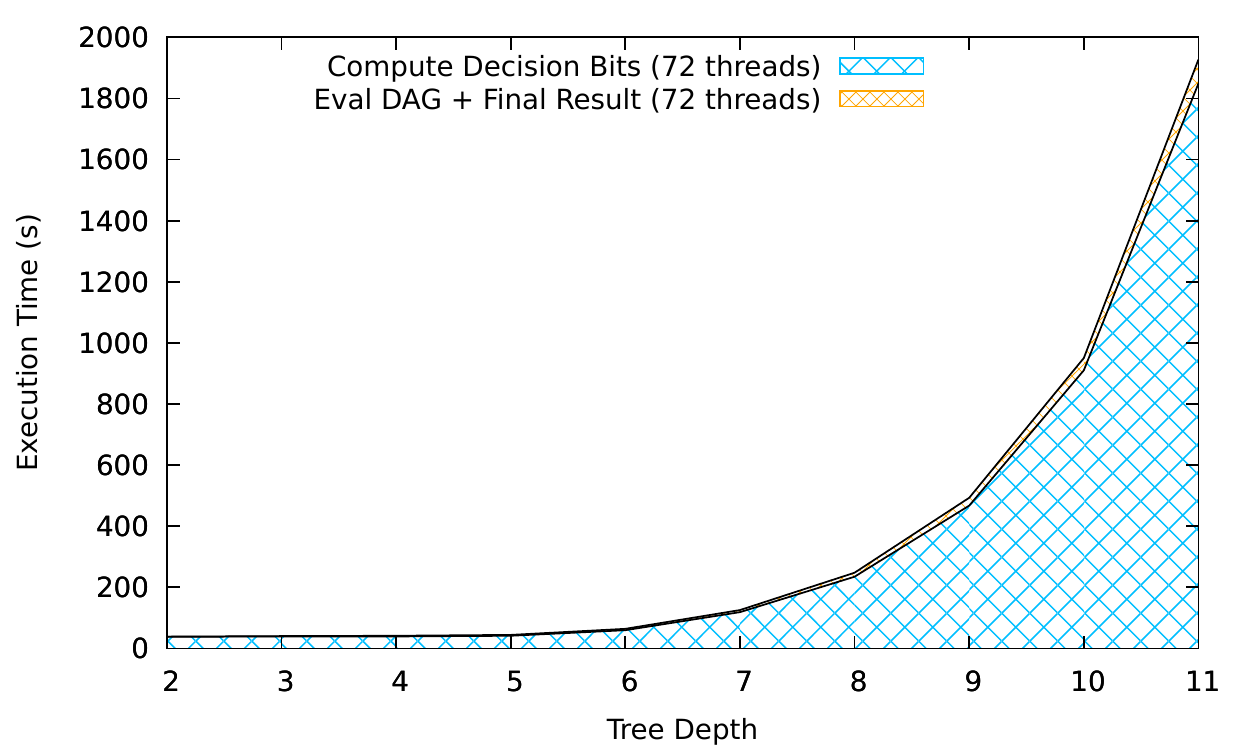}
		\caption{$\pdtebin$ Runtime with $\helibmed$  for 16-bit inputs} 
		\label{fig:relation_16bit}
	\end{center}
\end{figure}

\begin{figure}[tbp]
	\begin{center}
		\includegraphics[width=.4\textwidth]{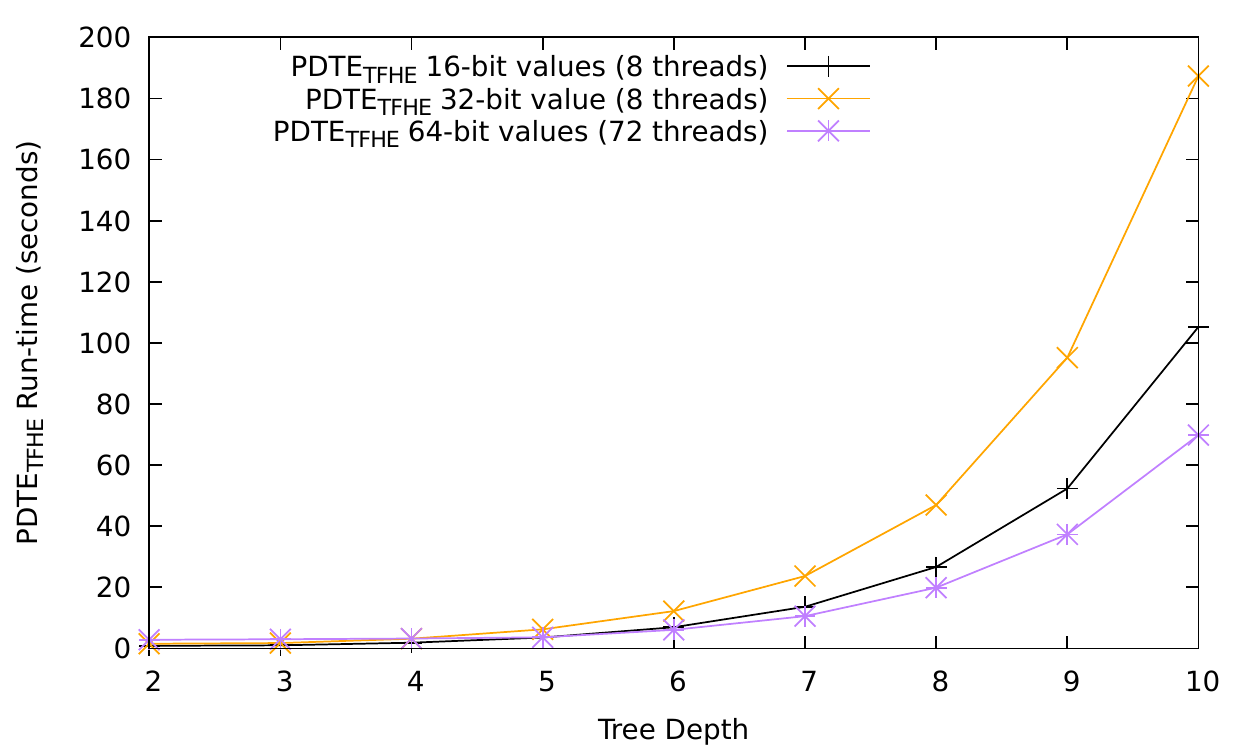}
		\caption{$\pdtebin$ Runtime with TFHE}
		\label{fig:pdte_tfhe}
	\end{center}
\end{figure}

\subsection{Performance of our Schemes on Real Datasets}
We also performed experiments on real datasets from the UCI repository \cite{UciRepository}. We performed experiments for both $\pdtebin$ and $\pdteint$ for the datasets illustrated in Table \ref{tab:our_vs_real} (parameters $n, d, m$ are defined in Table \ref{Notation_Table}). For $\pdtebin$, we reported the costs for HElib (single and amortized) and the costs for TFHE. Since TFHE evaluates only boolean circuits, we only have implementation and evaluation of $\pdteint$ with HElib. We also illustrate in Table \ref{tab:our_vs_real} the costs of two best previous works that rely only on homomorphic encryption, whereby the figures are taken from the respective papers \cite{LuZS2018, TaiMZC.2017}.
For one protocol run, $\pdtebin$ with TFHE is much more faster than $\pdtebin$ with HElib which is also faster than $\pdteint$ with HElib. However, because of the large number of slots, the amortized cost of $\pdtebin$ with HElib is better. For 16-bit inputs, our amortized time with HElib and our time with TFHE outperform XCMP \cite{LuZS2018} which used 12-bit inputs. For the same input bitlength, XCMP is still much more better than our one run using HElib, since the multiplicative depth is just 3. However, our schemes still have a better communication and $\pdtebin$ has no leakage. While the scheme of Tai et al.~\cite{TaiMZC.2017} in the semi-honest model has a better time for 64-bit inputs than our schemes for 16-bit inputs, it requires a fast network communication and at least double cost in the malicious model. The efficiency of Tai et al. is in part due to their ECC implementation of the lifted ElGamal \cite{ElGamal1985}, which allows a fast runtime and smaller ciphertexts, but is not secure against a quantum attacker, unlike lattice-based FHE as used in our schemes.


\begin{table}[tbp]
	\begin{center}
		\begin{tabular}{ c | c c c c }
			&Heart-disease & Housing  & Spambase & Artificial \\ 
			& (\heartdisease) & (\housing)  & (\spambase)& (\artificial) \\
			\hline 
			$n$ & 13 & 13 & 57 &  16 \\
			$d$ &  3 & 13 & 17 &  10\\
			$m$ &  5 & 92 & 58 & 500\\
		\end{tabular}
		\caption[Real Datasets and Model Parameters]{Real Datasets and Model Parameters}
		\label{tab:our_vs_real}
	\end{center}
\end{table}

\begin{table}[tbp]
	\begin{center}
		\begin{tabular}{l| c | c | c | c | c | c | c c}
			\multicolumn{1}{ l|}{} & \pdtebin &  \multicolumn{2}{ c|}{\pdtebin} & \multicolumn{2}{ c|}{\pdteint} & \cite{LuZS2018}& \cite{TaiMZC.2017}  \\ 
			\multicolumn{1}{ l|}{} & (TFHE) & \multicolumn{2}{ c|}{(HElib)}& \multicolumn{2}{ c|}{(HElib)} & (HElib)& (mcl) \\ 
			\multicolumn{1}{ l|}{$\secprm$}  & 128 & \multicolumn{2}{ c|}{150}  & \multicolumn{2}{ c|}{135} & 128 & 128\\
			\multicolumn{1}{ l|}{$\inputlen$}  & 16 & \multicolumn{2}{ c|}{16}  & \multicolumn{2}{ c|}{16} & 12 & 64\\  
			\multicolumn{1}{ l|}{\#thd}  & 16 & \multicolumn{2}{ c|}{16}  & \multicolumn{2}{ c|}{16} & 16 & -\\ \cline{2-8}
			\multicolumn{1}{ l|}{} & one & am. & one& am. & one & one & one \\ 
			\hline 
			\multicolumn{1}{ l|}{\heartdisease}  & 0.94 & 0.05 & 40.61 & 0.0073 & 45.59 & 0.59 & 0.25\\
			\multicolumn{1}{ l|}{\housing}    & 6.30 & 0.35 & 252.38 & 0.90 & 428.23 & 10.27 & 1.98 \\
			\multicolumn{1}{ l|}{\spambase}   & 3.66 & 0.24 & 174.46 & 0.72 & 339.60 & 6.88 & 1.80\\
			\multicolumn{1}{ l|}{\artificial}  & 22.39 & 1.81 & 1303.55 & 0.75 & 2207.13 & 56.37 & 10.42\\
		\end{tabular}
		\caption[Runtime of PDTE on Real Datasets]{Runtime (in seconds) of PDTE on Real Datasets: \scriptsize{$\secprm$ is the security level. $\inputlen$ is the input bit length. \#thd is the number of threads. mcl\cite{mcl} is a pairing-based cryptography library. Column \enquote{one} reports the time for one protocol run while \enquote{am.} reports the amortized time (e.g., the time for one run divided by $\slotnb$).}}
		\label{tab:our_vs_real}
	\end{center}
\end{table}


\section{Conclusion}
\label{Conclusion}
While almost all existing PDTE protocols require many interaction between the client and the server, we designed and implemented novel client-server protocols that delegate the complete evaluation to the server while preserving privacy and keeping the overhead low. Our solutions rely on SHE/FHE and evaluate the tree on ciphertexts encrypted under the client's public key. Since current SHE/FHE schemes have
high overhead, we combine efficient data representations with different algorithmic optimizations to keep the computational overhead and the communication cost low. 

\bibliographystyle{IEEEtranS}
\bibliography{IEEEabrv, NonInteractiveDT}

\appendices

\section{Security Analysis}
\label{security_proofs}
\subsection{Proof of Lemma \ref{LinTzengLemma}}
\begin{proof}
If $V = V_{x_i}^{1} - V_{y_j}^{0}$ has a unique  0 at a position $l, (1\leq l \leq \inputlen)$ then $u_{il}$ and $v_{il}$ have bit representation $z_{\inputlen-l+1} \cdots z_{1}$, where for each $h, \inputlen-l+1 \geq h \geq 2$, $z_{h} = x_{ig} = x_{jg}$ with $g=l+h-1$, and $z_{1} = x_{il} = 1$ and $x_{jl} = 0$. It follows that $x_i > y_j$.

\noindent If $x_i > y_j$ then there exists a position $l$ such that for each  $h, \inputlen \geq h \geq l+1$, $x_{ih} = x_{jh}$ and $x_{il} = 1$ and $x_{jl} = 0$. This implies $u_{il} = v_{il}$. 

For $h, \inputlen \geq h \geq l+1$, either $u_{ih}$ bit string is a prefix of $x_{i}$ while $v_{jh}$ is random, or $u_{ih}$ is random while  $v_{jh}$ bit string is a prefix of $y_{j}$. From the choice of $r^{(0)}_{ih}$, $r^{(1)}_{ih}$, we have $u_{ih} \neq v_{ih}$. 

For $h, l-1 \geq h \geq 1$ there are three cases: $u_{ih}$ and $v_{ih}$ (as bit string) are both prefixes of $x_{i}$ and $y_{j}$, only one of them is prefix, both are random. For the first case the difference of the bits at position $l$ and for the other cases the choice of $r^{(0)}_{ih}$ imply that $u_{ih} \neq v_{ih}$.
\end{proof}

\subsection{correctness}
The correctness for the basic scheme follows directly from Lemma \ref{thr:class_function_corretness}. For the binary implementation, we proved with Lemmas \ref{thr:log_mult_depth}, \ref{thr:dag_max_def}, \ref{thr:dag_mult_dep_graph} that aggregating the paths using Algorithms \ref{fig:eval_path_log_mult_depth} and \ref{alg:evalpathmultDAG} is correct. For the integer implementation, Lemma \ref{LinTzengLemma} ensures the correctness of the comparison. The classification path is marked with 0 on all edges while the other paths are marked with at least one random number. As a result, summing up the marks along the paths returns 0 for the classification path and a random number for all other paths.

\subsection{security}
It is straightforward to see that our protocols are secure. There is no interaction with the client during the computation and a semi-honest server sees only IND-CPA ciphertexts. A semi-honest client only learns the encryption of the result (and additional encryptions of random elements for $\pdteint$). A malicious server can only return a false classification result. This is inherent to private function evaluation where the function (the decision tree in our case) is an input to the computation. A malicious client can send a too \enquote{noisy} ciphertext, such that after the computation at the server a correct decryption is not possible, leaking some information. This attack works only with level FHE and is easy to deal with, namely the computation of a ciphertext capacity is a public function which the server can use to check the ciphertexts before starting the computation. Therefore, we state the following:
\begin{theorem}
Our protocols correctly and securely implement the PDTE functionality $\func{F}{PDTE}$.
\end{theorem}
As $\pdtebin$ returns the bit representation of the resulted classification label whose bitlength is public (i.e., the set of possible classification labels is known to the client), there is no leakage beyond the final output. $\pdteint$ returns as many ciphertexts as there are leaves and, therefore, leaks the number of decision nodes.

\section{Complexity analysis}
\label{complexity_analysis}
We now analyse the complexity of our scheme, distinguishing between the binary and the integer implementations. In the following, we assume that the decision tree is a complete tree with depth $d$.
\subsection{Complexity of the Binary Implementation}

The SHE comparison circuit has multiplicative depth $|\inputlen-1|+1$ and requires $\bigO{\inputlen \cdot |\inputlen|}$ multiplications \cite{CheonKK15, CheonKK16,CheonKL15}. That is, the evaluation of all decision nodes requires $\bigO{2^d \inputlen \cdot |\inputlen|}$ multiplications. The path evaluation has a multiplicative depth of $|d-1|$ and requires for all $2^d$ paths $\bigO{d2^d}$ multiplications. The evaluation of the leaves has a multiplicative depth of 1 and requires in total $2^d$ multiplications. The total multiplicative depth for $\pdtebin$ is, therefore, $|\inputlen-1|+|d-1| + 2 \approx |\inputlen|+|d| + 2$ while the total number of multiplications is $\bigO{2^d \inputlen \cdot |\inputlen| + d2^d + 2^d} \approx \bigO{d2^d}$.

For the label packing, the bit representation of each classification label is packed in one ciphertext. This hold for the final result as well. As a result, if the tree is complete and all classification labels are distinct, then the server sends $\ceil{\frac{d}{\slotnb}}$ ciphertext(s) to client. In practice, however, $\ceil{\frac{d}{\slotnb}} = 1$ holds as $d$ is smaller that the number $\slotnb$ of slots. 

For threshold packing, the decision bit at node $v$ will be encrypted as $ \ctxtrep{b_v | 0 | ... | 0}$. Then if we encrypt the classification label $c_i = c_{i|k|} ... c_{i1}$ as $\ctxtrep{c_{i|k|} | 0 | ... | 0}, ..., \ctxtrep{c_{i1} | 0 | ... | 0}$, the final result $c_l$ will be encrypted similarly such that with extra shifts, we can build the ciphertext $ \ctxtrep{c_{l|k|} | ... | c_{l1} | 0 | ... | 0}$. As a result, the server sends only 1 ciphertext back to the client.

For other cases (e.g., attribute packing, or no packing at all as in the current implementation of TFHE), the bits of a classification label are encrypted separately which holds for the final result as well. As a result the server sends back $d$ ciphertexts to the client.

\subsection{Complexity of the Integer Implementation}

The modified LinTzeng comparison circuit has multiplicative $|\inputlen - 1|$ and requires $\bigO{\inputlen - 1}$ multiplications. As a result, the evaluation of all decision node requires $\bigO{(\inputlen - 1)2^d}$ multiplications. In $\pdteint$, the path evaluation does not requires any multiplication. However, the leave evaluation has a multiplicative depth of 1 and requires in total $2^d$ multiplications. The total multiplicative depth for $\pdteint$ is therefore $|\inputlen-1| + 1 \approx |\inputlen| + 1$ while the total number of multiplications is $\bigO{(\inputlen - 1)2^d + 2^d} \approx \bigO{2^d}$.

For $\pdteint$, it is not possible to aggregate the leaves as in $\pdtebin$. If the client is classifying many inputs, the server must send $2^d$ ciphertexts back. If the client is classifying only one input, then the server can use shifts to pack the result in $\ceil{\frac{2^d}{s}}$ ciphertext(s).

\section{Homomorphic Operations in TFHE}
As already mentioned earlier, the current version of TFHE only supports binary gates. According to Chillotti et al.~\cite{ChillottiGGI17, ChillottiGGI18}, gate bootstrapping and gate evaluation cost about 13 ms for all binary gates except for the MUX gate, which costs 26 ms on a modern processor. For a full list of available gates, we refer to Chillotti et al.~\cite{TFHE}. In Table \ref{tab:tfhe_gates}, we illustrate the runtime of TFHE's gate evaluation with our testbed. The figures are given as average over 1000 runs.

\begin{table}[tbp]
	\begin{center}
		\begin{tabular}{l l c}
			Gate Name & Gate Functionality & Run-time (ms) \\ 
			& & 128-bit security\\
			\hline
			CONSTANT & \textit{result = encode(int)} & 0.00052 \\
			NOT & \textit{result} = $\neg a$ & 0.00051 \\
			COPY & \textit{result} = $a$  & 0.00035 \\
			NAND  & \textit{result} = $\neg(a \wedge b)$ & 11.32751\\
			OR  & \textit{result =} $ a \vee b$ & 11.40669 \\
			AND &  \textit{result =} $a \wedge b$& 11.38739  \\
			XOR &  \textit{result} = $a + b \bmod 2$  & 11.39326  \\
			XNOR &  \textit{result} = $(a = b)$ & 11.39418  \\						
			NOR & \textit{ result} = $\neg(a \vee b)$ & 11.39813  \\
			ANDNY & \textit{ result} = $\neg a \wedge b$ & 11.39255  \\
			ANDYN & \textit{ result} = $a \wedge \neg b$ & 11.39737  \\
			ORNY & \textit{ result} = $\neg a \vee b $ & 11.40777  \\
			ORYN & \textit{ result} = $ a \vee \neg b$ & 11.39940  \\
			MUX & \textit{ result} = $ a ? b : c$ & 21.29517  \\
		\end{tabular}
		\caption{TFHE Binary Bootstrapping Gates } 
		\label{tab:tfhe_gates}
	\end{center}
\end{table}

\section{Extension To Random Forest}
\label{Extension_Random_Forest}
In this section, we briefly describe how the binary implementation $\pdtebin$ can be extended to evaluate a random forest non-interactively.  A random forest is a generalization of decision tree  which consists of many trees. A classification with a random forest then evaluates each tree in the forest and outputs the classification label which occurs most often. Hence, the classification labels are ranked by their number of occurrences and the final result is the best ranked one.

Let the random forest consists of trees $\mathcal{T}_1, \ldots, \mathcal{T}_N$ and let $\pdtebin_S(\mathcal{T}_j, x)$ denote the evaluation of the decision tree $\mathcal{T}_j$ on input vector $x$ resulting in $\mathcal{T}_j(x) = R_j$, which is encrypted as  $\ctxtrep{\bitrep{R_j}} = (\ctxtrep{R_{j|k|}}, \ldots, \ctxtrep{R_{j1}})$, where $\bitrep{R_j} = R_{j|k|} \ldots R_{j1}$. Let's assume, there are $k$ classification labels $c_1, \ldots, c_k$ with $\bitrep{c_i} = c_{i|k|} \ldots c_{i1}$ and each $c_i$ has encryptions $\ctxtrep{\bitrep{c_i}} = (\ctxtrep{c_{i|k|}}, \ldots, \ctxtrep{c_{i1}})$ and  $\ctxtrep{\vec{c}_i} = \ctxtrep{c_{i|k|} | \ldots| c_{i1}}$. Let $f_i$ denote the number of occurrences of $c_i$ after evaluating the $N$ trees, with encryption $\ctxtrep{\bitrep{f_i}} = (\ctxtrep{f_{i|N|}}, \ldots, \ctxtrep{f_{i1}})$, where $\bitrep{f_i} = f_{i|N|} \ldots f_{i1}$. 

The computation requires the routine $\shecmp$ and two new ones: $\shefadder$ and $\sheequ$.

Full adder: Let $b_{i1}, \ldots, b_{in}$ be $n$ bits such that $r_i = \sum_{j=1}^{n}b_{ij}$ and let $r_i^b=r_{i\log n}, \ldots, r_{i1}$ be the bit representation of $r_i$. The routine $\shefadder$ implements a full adder on $\lsem b_{i1} \rsem, \ldots, \lsem b_{in} \rsem$ and returns $\lsem r_i^b \rsem = (\lsem r_{i\log n} \rsem, \ldots, \lsem r_{i1} \rsem)$. 

Equality testing: There is no built-in routine for equality check in HElib. We implemented it using \textsc{SheCmp} and  \textsc{SheAdd}. Let $x_i$ and $x_j$ be two $|k|$-bit integers.
We use \textsc{SheEqual} to denote the equality check routine and implement $\textsc{SheEqual}(\ctxtrep{\bitrep{x_i}},  \ctxtrep{\bitrep{x_j}})$ by computing:
\begin{itemize}
	\item $(\lsem b'_{i} \rsem$, $\lsem b''_{i} \rsem) = \textsc{SheCmp}(\ctxtrep{\bitrep{x_i}}, \ctxtrep{\bitrep{x_j}})$ and
	\item  $\lsem \beta_i\rsem = (\lsem b'_{i} \rsem \oplus \lsem b''_{i} \rsem \oplus \lsem 1 \rsem)$, which results in $\beta_i = 1$ if $x_i = x_j$ and $\beta_i = 0$ otherwise.
\end{itemize}

To select the best label, the random forest algorithm either uses majority voting or $\mathsf{argmax}$.
For majority voting, $c_i$ is the final result if an only if $f_i \geq t$, where $t = \ceil{\frac{N}{2}}$ with bit representation $\bitrep{t} = t_{|N|} \ldots t_{1}$ and encryption $\ctxtrep{\bitrep{t}} = (\ctxtrep{t_{|N|}}, \ldots, \ctxtrep{t_{1}})$. The computation is described in Algorithm \ref{fig:Random_Forest_PDTE}.
For $\mathsf{argmax}$, $c_i$ is the final result if an only if $f_i$ is larger than all other $f_j, j \neq i$. The computation is described in Algorithm \ref{fig:Random_Forest_PDTE_Max}.

\begin{figure}[tbp]
	\renewcommand{\figurename}{Algorithm}
	\centering
	\begin{mdframed}
	\begin{algorithmic}[1]
		
		\For {$j = 1$ \textbf{to} $N$} \label{start}
			\State $\ctxtrep{\bitrep{R_j}} \gets \pdtebin_S(\mathcal{T}_j, x)$ 
			\For {$i = 1$ \textbf{to} $k$}
				\State $\ctxtrep{b_{ij}} \gets \sheequ(\ctxtrep{\bitrep{R_j}}, \ctxtrep{\bitrep{c_i}})$
			\EndFor
		\EndFor
		\State $\ctxtrep{\result} \gets \ctxtrep{0}$
		\For {$i = 1$ \textbf{to} $k$}
				\State $\ctxtrep{\bitrep{f_{i}}} \gets \shefadder(\ctxtrep{b_{i1}}, \ldots, \ctxtrep{b_{iN}})$ \label{compute_fi}
				\State $\ctxtrep{e_{i}} \gets \shecmp(\ctxtrep{\bitrep{f_{i}}}, \ctxtrep{\bitrep{t}})$
				\State $\ctxtrep{\result} \gets \ctxtrep{\result} \oplus (\ctxtrep{e_{i}} \odot \ctxtrep{\vec{c}_i})$ 
		\EndFor
		\State \Return $\ctxtrep{\result}$
		
	\end{algorithmic}
	\end{mdframed}
	\caption[Private Random Forest Majority Voting]{Private Random Forest With Majority Voting}
	\label{fig:Random_Forest_PDTE}
\end{figure}

\begin{figure}[tbp]
	\renewcommand{\figurename}{Algorithm}
	\centering
	\begin{mdframed}
	\begin{algorithmic}[1]
		\State Compute $\ctxtrep{\bitrep{f_{i}}}$ as in Algorithm \ref{fig:Random_Forest_PDTE} Lines \ref{start} to \ref{compute_fi}
		\For {$i := 1$ \textbf{to}  $k$}
			\State $\lsem \beta_{ii} \rsem \leftarrow \lsem 1 \rsem$
			\For {$j := i+1$ \textbf{to}  $k$} 
				\State $(\lsem \beta_{ij} \rsem, \lsem \beta_{ji} \rsem) \leftarrow \textsc{SheCmp}(\ctxtrep{\bitrep{f_{i}}}, \ctxtrep{\bitrep{f_{j}}})$
			\EndFor
		\EndFor
		\For {$i := 1$ \textbf{to}  $k$} 
				\State $\lsem \bitrep{r_{i}} \rsem \leftarrow \textsc{SheFadder}(\lsem \beta_{i1} \rsem, \ldots, \lsem \beta_{ik} \rsem)$
		\EndFor
		\For {$i := 1$ \textbf{to}  $k$}
			\State $\lsem e_{i} \rsem \leftarrow \textsc{SheEqual}(\lsem \bitrep{r_{i}} \rsem, \lsem \bitrep{k} \rsem)$ \label{compute_f_r}
		\EndFor
		\For {$i := 1$ \textbf{to}  $k$}
			\State $\ctxtrep{\result} \gets \ctxtrep{\result} \oplus (\ctxtrep{e_{i}} \odot \ctxtrep{\vec{c}_i})$
		\EndFor
		\State \Return $\ctxtrep{\result}$
	\end{algorithmic}
	\end{mdframed}
	\caption[Private Random Forest Maximum Voting]{Private Random Forest with Maximum Voting}
	\label{fig:Random_Forest_PDTE_Max}
\end{figure}

\end{document}